\documentclass[sigconf
]{acmart}
\newtheorem*{untheorem}{Theorem}
\newtheorem*{corollary 4.5}{Corollary 4.5}
\newtheorem*{lemma 5.8}{Lemma 5.8}
\newtheorem*{lemma 6.2}{Lemma 6.2}
\newtheorem*{theorem 6.5}{Theorem 6.5}
\newtheorem*{theorem 7.6}{Theorem 7.6}
\newtheorem*{theorem 7.9}{Theorem 7.9}
\newtheorem{remark}{Remark}
\usepackage{graphicx}
\usepackage{subfigure}
\usepackage{amsthm}
\usepackage{amsmath}

\AtBeginDocument{%
  }

\setcopyright{acmlicensed}
\copyrightyear{2018}
\acmYear{2018}
\acmDOI{XXXXXXX.XXXXXXX}
\acmConference[]{}{}{}
\acmISBN{}

\begin{document}

\title{Automatic constraint satisfaction problem}


\author{Xiaoyang
Gong}
\affiliation{%
  \institution{University of Electronic Science and Technology of China}
  \city{Chengdu}
  \state{Sichuan}
  \country{China}}
\email{XiaoyangGong\_uestc@outlook.com}

\author{Xinyao Wang}
\affiliation{%
  \institution{University of Electronic Science and Technology of China}
  \city{Chengdu}
  \state{Sichuan}
  \country{China}}
\email{xinyao@std.uestc.edu.cn}

\author{Bakh Khoussainov}
\affiliation{%
  \institution{University of Electronic Science and Technology of China}
  \city{Chengdu}
  \state{Sichuan}
  \country{China}}
\email{bmk@uestc.edu.cn}

\author{Andrei Bulatov}
\affiliation{%
  \institution{Simon Fraser University}
  \city{}
  \state{}
  \country{Canada}}
\email{abulatov@sfu.ca}


\newcommand{\Ac}{\mathcal{A}}
\newcommand{\Bc}{\mathcal{B}}
\newcommand{\Cc}{\mathcal{C}}
\newcommand{\Dc}{\mathcal{D}}
\newcommand{\Fc}{\mathcal{F}}
\newcommand{\Lc}{\mathcal{L}}
\newcommand{\Pc}{\mathcal{P}}
\newcommand{\Kc}{\mathcal{K}}

\newcommand{\Nb}{\mathbb{N}}

\newcommand{\ub}{\mathbf{u}}

\newcommand{\Nat}{\mathbb{N}}
\newcommand{\ar}{\rightarrow}
\newcommand{\partialar}{\rightharpoonup}

\newcommand{\ccl}[1]{\langle #1\rangle}
\newcommand{\fo}[1]{\underline{#1}}
\newcommand{\Age}[1]{\text{Age}(#1)}

\def\bs{\mathbf{s}}
\def\CSP{\mathrm{CSP}}
\def\AutCSP{\mathrm{AutCSP}}
\def\Pol{\mathrm{Pol}}
\def\Inv{\mathrm{Inv}}
\let\vf=\varphi

\newcommand{\anote}[1]{\textcolor{blue}{(Andrei: #1)}}
\newcommand{\tgr}[1]{\textcolor{green}{#1}}
\newcommand{\tyellow}[1]{\textcolor{yellow}{#1}}

\newcommand{\xynote}[1]{\textcolor{orange}{Xiaoyang: #1}}
\newcommand{\wnote}[1]{\textcolor{orange}{Xinyao: #1}}
\newcommand{\bakhnote}[1]{\textcolor{red}{Bakh: #1}}
\newcommand{\algoquad}{\hspace*{6mm}}

\begin{abstract}
We study constraint satisfaction problems (CSPs) where the constraint languages  are defined by finite automata, giving rise to automata-based CSPs.
The key notion is the concept of Automatic Constraint Satisfaction Problem ($\AutCSP$) where constraint languages and instances are specified by finite automata. The $\AutCSP$ captures infinite yet finitely describable sets of relations, enabling concise representations of complex constraints. Studying the complexity of the $\AutCSP$s illustrates the interplay between classical CSPs, automata, and logic, sharpening the boundary between tractable and intractable constraints. We show that checking whether an operation is a polymorphism of such a language
can be done in polynomial time. Building on this, we establish several complexity classification results for the $\AutCSP$. In particular, we prove that Schaefer’s Dichotomy Theorem extends to the $\AutCSP$ over the Boolean domain, and we provide algorithms that decide tractability of some classes of $\AutCSP$s over arbitrary finite domains via automatic polymorphisms. An important part of our work is that our polynomial-time algorithms run on $\AutCSP$ instances that can be exponentially more succinct than their standard CSP counterparts.
\end{abstract}




\maketitle

%\anote{-- In order to separate edits and comments I made edits green and comments blue\\ -- I tried to improve the intro. Feel to change whatever you want.\\ -- Add statements we prove in the Appendix to the respective section of the appendix.\bakhnote{Xinyao does it.}\wnote{done}\\ -- Apart from comments in the text, pls, check the consistency of terminology and notation. I caught some issue, but there are probably more. Annoyed reviewer is the last thing we need.}

\section{Introduction}\label{S:Intro}

The classical Constraint Satisfaction Problem (CSP) asks whether there exists an assignment to a set of variables that satisfies all given constraints. The CSP has found its applications across many areas of mathematics, computer science, and artificial intelligence. Many of its variants received much attention in the last several decades. In this paper we are concerned with the decision version of the problem, especially its `non-uniform' kind, in which the type of allowed constraints is restricted by a collection of predicates often called a constraint language.

The study of non-uniform CSPs was initiated by the seminal paper of Feder and Vardi \cite{Feder98:monotone}. A major turning point was the introduction of the algebraic approach to CSPs by Jeavons, Cohen, and Gyssens \cite{Jeavons97:closure,Jeavons98:algebraic} and by Bulatov, Jeavons, and Krokhin \cite{Bulatov05:classifying}. The study of the complexity of non-uniform CSPs has been a major research direction spanning four decades starting from the work by Schaefer \cite{SchaeferDic}, through many intermediate results \cite{Bulatov06:3-element,Bulatov11:conservative,Barto14:constraint,Idziak10:tractability}, and finally a complete characterization of the complexity of non-uniform CSPs proved independently by Bulatov \cite{DBLP:conf/focs/Bulatov17} and Zhuk \cite{DBLP:journals/jacm/Zhuk20}. Beyond the finite-domain setting, the study of CSPs over infinite domains was initialized by Bodirsky and Ne\v set\v ril \cite{Bodirsky03,Bodirsky_PhdThesis}
and came to a new layer of complexity and model-theoretic challenges %, as surveyed in 
\cite{pinsker2022current}.

We aim to explore constraint satisfaction problems, in which constraint languages are described by finite automata. Our automata-based formulation of the CSPs naturally extends the classical non-uniform CSPs to the settings with infinite but finitely representable constraint languages. In particular, we initiate a program of studying the dichotomy problem for automata-based CSPs. 
%% Since the paper connects two distinct research areas, -- the constraint satisfaction problem and automata theory, we took a liberty to expose some basic concepts and results from both of these areas in a more gentle way in different parts of this work.

In order to put some meat on the bones, we recall the formal definition of the CSP. A \textbf{ CSP instance} is a triple $I=(V,D, \Cc)$, where $D$ is a finite non-empty domain, $V$ is a set of variables, and $\Cc$ is a finite set of constraints. Each constraint is a pair $\ccl{\bs,R}$, where $\bs=(s_1, \ldots,s_k)$ is a tuple of variables from $V$  and $R\subseteq D^k$ is a $k$-ary relation over $D$. A \textbf{solution} of $I$ is a function $\vf:V\to D$ such that for each constraint $\ccl{\bs,R}\in\Cc$, $\bs=(s_1,\dots, s_k)$, it holds that $(\vf(s_1),\dots,\vf(s_k))\in R$.   An alternative view of the instance $I$ is the conjunction of formulas of the form $R(s_1, \ldots,s_k)$.

The CSP is NP-complete, as it contains known NP-complete problems. Hence, the problem is often restricted either by imposing conditions on the way the variables % constraint scopes 
interact, or by limiting the relations that appear in constraints. The latter is known as the non-uniform CSP parametrized by a \textsc{constraint language} $\Gamma$, -- a set (finite or infinite) of relations on the domain $D$. The problem $\CSP(\Gamma)$ includes all instances in which constraint relations belong to the language $\Gamma$. The standard examples of non-uniform CSPs 
include the well-known \textsc{Satisfiability} problem, in which the constraint language consists of relations over $\{0,1\}$ expressible as clauses; the \textsc{Graph $k$-Coloring} problem, which can be translated into a CSP by using the constraint language consisting solely of the disequality relation on a $k$-element set; and \textbf{Linear Equations}, which encodes solving systems of linear equations through a CSP whose constraints are collections of solutions of a linear equation. These examples illustrate the overall research goal of delineating hard and easy non-uniform CSPs.

Although the complexity of non-uniform CSPs has been characterized in \cite{DBLP:conf/focs/Bulatov17,DBLP:journals/jacm/Zhuk20}, the algorithms in these papers assume a specific representation of instances. The choice of representation may make a dramatic difference in terms of the running time of algorithms and the complexity of the problems. This especially matters when dealing with infinite constraint languages. Note that the constraint languages in the  \textsc{Satisfiability} and \textsc{Linear Equation} problems are infinite, while in the \textsc{Graph $k$-Coloring} problem, it is finite. The standard conventions to represent instances of the CSP are:

--\textbf{Uniform CSP representation}: In this case the representation of the instance $I=(V, D, \Cc)$ is \textbf{explicit}, that is, it lists all the variables in $V$, all domain elements in $D$, and all tuples that appear in relations that occur in $\Cc$. Hence,  the size of $I$ is $|V|+|D|+|\Cc|$, where $|\Cc| = \sum_{\ccl{\bs, R}\in \Cc} ( |\bs| \cdot |R|)$.

-- \textbf{Non-uniform CSP($\Gamma$) representation, $\Gamma$ is finite}: The language $\Gamma$ is a parameter  and is not part of input instances. For a given instance $I=(V, D, \Cc)$,
constraints $\ccl{\bs,R} \in \Cc$ in the instance are represented by the symbol $R$ and $\bs$. Hence, the size of $I$ is  $|V|+|\Cc|$, 
where $|\Cc| = \sum_{\ccl{\bs, R}\in \Cc} |\bs| $.

-- \textbf{Non-uniform CSP($\Gamma$) representation, $\Gamma$ is infinite}: It is standard in the literature on non-uniform CSPs to represent relations \textbf{explicitly} %, just as the uniform case,  
by a complete list of tuples they contain. 

Variations of these conventions are used in some cases. For example, in the \textsc{Satisfiability} problem, it suffices to indicate the literals occurring in a clause and their polarities, which is exponentially more concise than listing all the tuples satisfying the clause. In a similar way \textsc{Linear Equations} do not require listing all solutions of a linear equation. Both problems admit a representation exponentially more concise than that of the explicit one. Sometimes CSPs over an infinite language are replaced with an `equivalent' CSP over a finite one. For example, rather than considering the \textsc{Satisfiability} problem, one considers the \textsc{3-Satisfiability} problem. A similar phenomenon occurs in \textsc{Linear Equations}, in which case a linear equation is replaced with several linear equations containing at most 3 variables; some auxiliary variables may be added as well.

If a CSP admits some algebraic structure, as in the case of \textsc{Linear Equations}, yet another way of representing constraints is possible.

 --\textbf{Representation by a set of generators}: 
This is studied in \cite{Bulatov2006SimpleMaltsev,Dalmau06:generalized,Idziak10:tractability}, where
rather than listing all tuples in a relation, one finds a compact representation of instances via generators. 
Constraint languages with such compact representations have an edge polymorphism \cite{Berman10:constraint}. While compact representations help design solution algorithms when the constraints are given explicitly, it is unclear if they can also be helpful to  represent constraint relations. Indeed, building on the ideas of Berkholz and Grohe \cite{Berkholz15:limitations}, an algorithm capable of handling constraints with an edge polymorphism given by their compact representations would improve solving the \textsc{Graph Isomorphism} problem.

In this paper we study the complexity of CSPs, in which constraints are represented in some concise way. Clearly, the ultimate concise representation would be by a circuit that tests whether or not a tuple is accepted by a constraint. However, this seems to be too general, as even checking if a constraint accepts any tuple at all is NP-complete. Hence, a more tame framework would clearly be more of  interest. In this paper, we propose such a framework  based on finite automata. More precisely, we consider the \textbf{Automatic Constraint Satisfaction Problem} (AutCSP for short), whose instances are  quadruples of the form $I=(V,D,\Cc,\Ac)$, where $V$ is a set of variables, $D$ a domain, $\Ac$ a finite automaton with domain $D$, and $\Cc$ a set of constraints of the form $\langle \bs,R\rangle$. Here, $\bs$ is a tuple of variables (of length $|\bs|$) and $R=\mathcal L(\Ac)\cap D^{|\bs|}$, where $\mathcal L(\Ac)$ is the language recognized by $\Ac$. If the automaton $\Ac$ is fixed, we denote the corresponding problem AutCSP($\Ac$). The overall goal is to determine the complexity of AutCSP($\Ac$) depending on the automaton $\Ac$.

We now briefly describe our paper. We define the automatic constraint satisfaction problem, give examples of automatic CSPs,
recast polymorphisms in the $\AutCSP$ setting, and prove some of their basic properties. We also introduce automatic constraint languages $\Gamma_{\Ac}$ that contain all the relations admissible in instances of AutCSP($\Ac$), and observe that AutCSP($\Ac$) is at least as difficult as the problem CSP($\Gamma_{\Ac}$), %\anote{We should have such a statement explicitly, I think.} 
which means that we need to focus on algorithmic results about the AutCSP. Then, to help with this project, we turn to polymorphisms of constraint languages that are exceptionally important in the study of the classical CSP. We design algorithms for automata whose corresponding constraint language has one of a number of types of polymorphisms, such as polymorphisms on the Boolean domain $\{0,1\}$ identified in Schaefer's Dichotomy Theorem \cite{SchaeferDic}, affine operations $x-y+z$, where $+$ and $-$ are operations of a finite field, and near-unanimity operations. We also prove that it is polynomial time to recognize, given an automaton $\Ac$ and an operation $f:D^k\to D$, whether or not $f$ is a polymorphism of $\Gamma_{\Ac}$. Furthermore, we discuss width 1 automatic constraint languages. Finally, we prove Schaefer's Dichotomy Theorem in the setting of $\AutCSP$. 
A detailed account of our contributions is in Section~\ref{sec:contributions} after we fully provide the necessary concepts in the next section. 
%\anote{I'd move it here}}

\section{Preliminaries}

Since the paper connects two distinct research areas, -- the constraint satisfaction problem and automata theory, we took the liberty to expose some basic concepts and results from both of these areas in a more gentle way in different parts of this work. For instance, in the CSP case, we  present some known algorithms such as the 1-minimality algorithm and the decomposition of relations closed under the majority operation into binary relations; in the automata case, we explain finite automata and automatic relations borrowed from the theory of automatic structures. A hope is that this exposure will be appreciated with the progression of this paper.  Also, in the paper we use the following {\bf conventions}: 

%We introduce automatic constraint languages, automatic constraint satisfaction problem $\AutCSP$, and discuss representations of the $AutCSP$ instances. Then we outline basic properties of polymorphisms in the setting of $\AutCSP$ and provide examples. 

{\bf 1.} For the most part of this paper we fix a finite domain $D$, where $|D|>1$.  In Section \ref{S:BooleanDomain}, the domain $D$ will be  Boolean $D=\{0,1\}$.  

 {\bf 2.} We set $[n]=\{1,2,\ldots, n\}$.  Any $n$-tuple $t\in D^n$ is a function $t:[n] \ar D$. For a set $I = \{i_1,\ldots, i_k\}\subseteq [n]$, the projection of $t$ onto $I$, written $\pi_I t$, is the restriction $\hat t:I\ar D$ of $t$.
%i.e. $(t(i_1),\ldots, t(i_k))$.  
The projection of $R\subseteq D^n$ onto $I$ is $\pi_I R = \{\pi_I   r:   r \in R\}$. Given a tuple $t:[n]\to D$, we may treat tuple $t$ as a sequence $\{t(i): 1\le i \le n\}$ as well. In this case, $|t|$ refers to the length of $t$ (or the cardinality of its domain). 
%\anote{Do we use bars for tuples or not? It's not consistent at the moment. \bakhnote{Xiaoyang will fix it.} Also, presumably, $|t|$ is supposed to denote the cardinality of the domain of $t$} 

{\bf 3.} Typically we denote languages of strings by $\mathcal L$ or $\mathcal L (\mathcal A)$, and languages consisting of relations by $\Gamma$ or $\Gamma_{\mathcal A}$.

\subsection{Automatic relations and operations}\label{S:Automatic-Relations}

We borrow terminology and several results from automatic structures \cite{Blumensath_Gradel_LICS2000_Automatic_Structures,1995BakhNerode}. 
%We use automata for defining automatic constraint languages. 
Let $D^*$ be the set of finite words over $D$.

\begin{definition}
A \textbf{finite automaton} $\Ac$ is a tuple $(S,I, \Delta, F)$, where $S$ is a finite set of states, $I\subseteq S$ is the set of initial states, $\Delta\subseteq S\times D \times S$ is the transition table, and $F\subseteq S$ is the set of accepting states. The automaton $\Ac$ is  \textbf{deterministic} if $\Delta$ is a function from $S\times D$ to $S$ and $I$ is a singleton.
\end{definition}

The automaton $\Ac$ is usually represented as a directed labeled graph with vertex set $S$. We put an edge from  state $p$ to $q$ and label the edge by $d\in D$ if $(p,d,q)\in \Delta$. Let $w=d_1, \ldots, d_n \in D^*$ be a word. A \textbf{run} of $\Ac$ on $w$ is a sequence $s_1, \ldots, s_{n+1}$ of states such that $s_1\in I$ and $(s_i, d_i, s_{i+1})\in \Delta$ for all $i=1, \ldots, n$.  The run $s_1, \ldots, s_{n+1}$ is \textbf{accepting} if $s_{n+1}\in F$. The word $w$ \textbf{is accepted} by $\Ac$ if there is an accepting run of the automaton on $w$. Thus, acceptance of the word $w$ is equivalent to the existence of a path in the graph representation of $\Ac$ such that the path is labeled by $w$, it starts at an initial state, and ends in an accepting state. This view allows for a straightforward solution of several important tasks:

-- The emptiness problem for finite automata, that is finding out if a given automaton accepts at least one string, can be decided in linear time by checking whether any of the accepting states are reachable from any of the initial states.

-- Checking if there is a word of a specific length $n$ accepted by $\Ac$. Again, the BFS enhanced with keeping track of vertices reachable by a path of a given length allows one not only to do that, but also to find such a word if one exists. 

-- The infinity problem of checking if a given automaton accepts infinitely many strings, can be solved in linear time. 
%.  This can also be decided in linear time by detecting an accepting state $p$ reachable from and an initial state such that there is a path from $p$ to $p$. 

%\smallskip
Let us define  $\mathcal L(\Ac)=\{w\in D^*\mid \Ac \text{ accepts }  w\}$. This set is called \textbf{a regular} (or equivalently,  finite automata recognizable) language. The class of regular languages is closed under the Boolean operations of union, intersection, and complementation \cite{10.1147/rd.32.0114}. 

We need the concept of automatic (regular) relations on $D^*$. We  first 
explain this concept on binary relations $R$. \ Let $x=x_1\ldots x_n$ and $y=y_1\ldots y_m$ be strings. We code the pair $(x,y)$ as the string  $z=z_1 \ldots z_{\max\{m,n\}}$ in the following manner:   
(1)  $z_i=(x_i, y_i)$ if $i\leq \min\{n,m\}$,   (2) $z_i=(\diamond, y_i)$ if $n<i\leq m$, and  (3)  $z_i=(x_i, \diamond)$ if $m<i\leq n$, where $\diamond\notin D$.
We call $z$ the \textbf{convolution} of $x$ and $y$ and denote it by $z=x\oplus y$. The string $z=x\oplus y$ is a string over the new alphabet
$D'= \{(\sigma, \diamond)\mid \sigma \in D\}  \cup  
\{(\diamond, \sigma)\mid \sigma \in D\}  \cup  D\times D$. 
For the binary relation $R$, consider the language $\oplus(R)$ over $D'$ consisting of all convolutions of pairs $(u,v)\in R$.  The relation $R$ is  \textbf{automatic (or regular)}  if there exists a finite automaton $\Ac$ recognizing $\oplus(R)$. The automaton $\Ac$ runs on a pair $(x,y)$ by reading symbols of $x$ and $y$ from left to right synchronously.   This definition can easily be generalized to $n$-ary relations $R$ on $D^*$. The alphabet $D'$ of the language $\oplus (R)$ consists of tuples of the form $(\sigma_1, \ldots, \sigma_n)$, where each $\sigma_i \in D\cup \{\Diamond\}$. Hence, we have the following  definitions: 

\begin{definition}\cite{1995BakhNerode,Blumensath_Gradel_LICS2000_Automatic_Structures}[\textbf{Automatic relations and operations}]\label{Dfn:Automatic-Relation}
A relation $R$ on the domain $D^*$ is \textbf{automatic (or regular)} if  the set $\oplus(R)$ is recognized by a finite automaton. A $k$-ary (partial) operation $f$ on the domain $D^*$ is  \textbf{automatic} if the graph of $f$,    
$$
Graph(f)=
\{((x_1, \ldots, x_k), f(x_1,\ldots, x_k)) \mid x_1, \ldots, x_k \in D^*\},
$$
is an automatic relation. 
%\anote{pls, use the same word map/function/operation for the same type of mappings}\bakhnote{Xinyao} \wnote{done. comment: we call everything of the form $D^n \to D$ a map, such as all polymorphisms. everything else will be called a function, such as our patterns, tuples($[n] \to D$), solutions($V \to D$), and transition functions of automata.}
\end{definition}
\begin{example}
    The following are  automatic relations on $D^*$:
    \begin{enumerate}
        \item The length equality predicate $\{(x,y) \mid |x|=|y|\}.$
        \item The length comparison $x \le_{len} y$ iff  $|x |\le |y|$.
        \item The relation $x\leq_{pref} y$ iff $x$ is a prefix of $y$.
        \item The operation $S_d(x)=xd$, where $x\in D^*$ and $d\in D$. 
        \item The  lexicographical linear order $\leq_{lex}$ on $D^*$.  
    \end{enumerate}
    The automata for examples (2) and (3) are shown in Figure~\ref{fig:automata} below.
\end{example}
\begin{figure}[htbp]
    \centering
    \subfigure[Automaton for $\le_{len}$]{
        \includegraphics[width=0.19\textwidth]{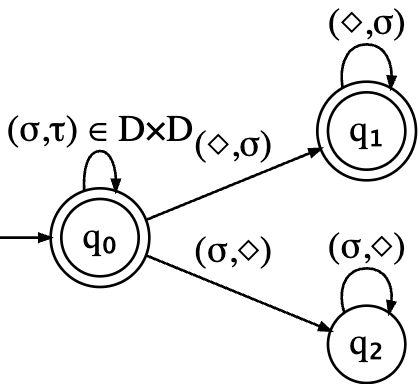}
        \label{fig:aut1}
    }
    \hspace{0.01\textwidth}  % 间距
    \subfigure[Automaton for $\leq_{pref} $]{
        \includegraphics[width=0.21\textwidth]{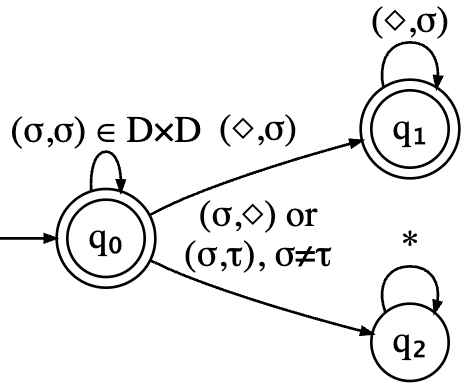}
        \label{fig:aut2}
    }
    \caption{Automata for $\le_{len}$ and  $\leq_{pref} $}
    \label{fig:automata}
\end{figure}
A first order structure $\mathcal M=(M; R_1, \ldots, R_k, f_1, \ldots, f_n)$ is \textbf{ automatic } if its domain $M$, relations $R_1, \ldots, R_k$ and the operations $f_1, \ldots, f_n$ are all automatic. The final ingredient needed from the theory of automatic structures is the next theorem. 

\begin{theorem}[\textbf{ Decidability Theorem}, \cite{Blumensath_Gradel_LICS2000_Automatic_Structures,1995BakhNerode}]\label{DecidabilityThm} 
    There exists an algorithm that, given an automatic structure $\mathcal M$ and a first order formula $\phi(x_1, \ldots, x_n)$, builds an automaton $\Ac_{\phi}$ such that for all $a_1, \ldots, a_n\in M$ we have $\mathcal M \models \phi(a_1, \ldots, a_n)$ if and only if $\Ac_{\phi}$ accepts  $(a_1, \ldots, a_n)$. In particular, the first order theory of any automatic structure $\mathcal M$ is decidable. 
\end{theorem}

\subsection{AutCSP and AutCSP($\Ac$) }\label{SS:size-aut-constraint}

In this section, we introduce automatic constraint satisfaction problems. In defining these problems a major consideration is the representation of automatic constraint instances and their sizes. We start with the definition of automatic constraint languages. 
%Recall that we have a fixed domain $D$. 
 
\begin{definition}\textbf{(Automatic constraint language)}\label{Dfn:AutCSP} \ 
Let $\Ac$ be a finite automaton over $D$. Define the following  set of relations $\Gamma_{\Ac} = \{R_1, \  R_2,  \  \ldots \}$,  where $R_k=\mathcal L(\Ac) \cap D^k $ and $k \in  \Nat$. We call a language of relations $\Gamma$ an \textbf{automatic constraint language} if  $\Gamma=\Gamma_{\Ac}$ for some finite automaton $\Ac$.
 \end{definition}

Now we define automatic CSP instances:
 \begin{definition}\textbf{ (AutCSP instance)} \ 
 We call the tuple $(V, D, \Cc, \Ac)$  an \textbf{AutCSP instance $I$} or \textbf{ an automatic CSP instance}, where 
\begin{enumerate}
    \item $V$ is a set of variables,  
    \item $D$ is a finite domain, 
    \item $\Cc$ is a finite set of constraints, and  
    \item $\Ac$ is a finite automaton over $D$.  
\end{enumerate}
Here a \textbf{constraint} is  a tuple $\ccl{\bs}=(s_{i_1}, \ldots,s_{i_k})$ of variables from $V$ that in the context of constraint satisfaction instance is viewed as the formula $R_k(s_{i_1}, \ldots, s_{i_k})$. \  Sometimes we write 
the constraint $\ccl{\bs}$ as $\ccl{\bs, R_k}$, -- consistent with the notations in the CSP literature,  where $R_k$ is a relation symbol of arity $k$.
%\footnote{There is another natural way to represent constraints using finite automata: Every constraint is specified through a sequence of variables \textbf{and} its own finite automaton that determines which tuples are accepted by the constraint. This representation clearly has more expressive power and still allows for a more succinct representation of constraints. However, this direction is outside of the scope of this paper.}
\end{definition}
 
%Let $R_k$ be the relation corresponding to the constraint  $\ccl{\bs}$.  Then any tuple $(d_1, \ldots, d_k)$ from $D^k$ belongs to $R_k$ if and only if the tuple  $(d_1, \ldots, d_k)$ is accepted by the automaton $\Ac$ as string of length $k$. Hence, if we want to represent the relation $R_k$, then instead of representing $R_k$ explicitly, where we list all tuples in $R_k$, we can represent $R_k$ implicitly by the  $k$-tuple $\ccl{s_1, \ldots, s_k}$ of variables. \anote{The last 2 paragraphs are a bit repetitive, I'll try to rewrite}
%The size of this representation of $R_k$ is $|\bs|$. This representation can be  exponentially smaller than the explicit representation of the relation $R_k$.

\textbf{ Automatic representations:}  Since representations of $\AutCSP$ instances are important, we need to explain the sizes of $\AutCSP$ instances. Recall that $|X|$ is the cardinality of $X$ (if $X$ is a set) or the length of $X$ (if $X$ is a tuple).  The size of a finite automaton $\Ac=(S, q_0, \Delta, F)$, written as $|\Ac|$, is  $|S|+|\Delta|$. 

%The instance consists of two parts. The first is the standard CSP instance $I'=(V, D, \Cc)$ and the second is the  automaton $\Ac$. 
%Therefore, we have the following:

\begin{definition}
 Let $I=(V, D, \Cc, \Ac)$  be an automatic CSP instance.
The \textbf{ size of the instance $I$} is defined in two ways: 
\begin{enumerate}
      \item \textbf{ The automaton is a part of the input:} The size of the representation of $I$ is $|V|+|\Cc|+|\Ac|$,  where  $|\Cc| = \sum_{\ccl{\bs, R}\in \Cc} |\bs| $.

    \item \textbf{ The automaton is a fixed parameter:} 
    The size of the representation of $I$ is $|V|+|\Cc|$,  where $|\Cc|$ is defined as above. So, $\Ac$ is not part of the input.
\end{enumerate}
\end{definition}

%\begin{remark}\label{rem:representations}

It is important to stress that the main difference between the classical CSPs and automatic CSPs is in their representations. The  AutCSPs instances are more succinct than their classical counterparts.  Representing the constraint $\ccl{\bs,R_k}$, in the $\AutCSP$ case, needs $|\bs|$ units of space, while 
representing it \textbf{explicitly} in the standard CSP setting requires $\sum_{r \in R_k}|r|=|R_k|\cdot |\bs|$ units of space. 
%In the classical CSP  case, representing the constraint $\ccl{\bs,R_k}$ explicitly requires $|\bs|$ units of space  only if $R_k$ is a singleton set. 
For any automatic constraint language $\Gamma_{\Ac}$, the function $gr_{\Ac}(k)=|R_k|$, where $R_k\in \Gamma_{\Ac}$, is always a polynomial or exponential \cite{DBLP:journals/tcs/Flajolet87}. Hence, our presentations of the $\AutCSP$ instances can be exponentially more succinct than their explicit representations. 
We also call our representations  \textbf{automatic representations}.

A \textbf{ solution} to an $\AutCSP$ instance $I=(V, D, \Cc, \Ac)$ is a function $\varphi : V \ar D$ such that for each constraint 
$\ccl{\bs}$, the tuple $\varphi(\bs)$ belongs to $R_{|\bs|}$. Thus, $\varphi(\bs)$ is a solution iff $\varphi(\bs)$ is accepted by the automaton $\Ac$. The goal is to design an algorithm that, given an $\AutCSP$ instance  $I$, decides if $I$ is satisfiable (that is, has a solution).

\begin{definition}
Based on the representations of $\AutCSP$ instances, we  define the following classes of CSP.
\begin{itemize}
\item \textbf{ The automatic constraint satisfaction problem}, written \textbf{ AutCSP}, 
is the problem of finding solutions to instances $I=(V, D, C, \Ac)$ where the automaton $\Ac$ is counted as a part of the input.   
\item The constraint satisfaction problem that corresponds to $\Gamma_{\Ac}$ is denoted by \textbf{ AutCSP($\Ac$)}. For  the AutCSP($\Ac$) instances, the automaton $\Ac$ is a fixed parameter and it is not part of the input. Hence we can write these instances in the form $(V, D, \Cc)$ as
$\Gamma_{\Ac}$ is  given by  $\Ac$. 
%Sometimes we call this \textbf{ automatic constraint satisfaction problem with the parameter $\Ac$.}
\end{itemize}
\end{definition}

It is easy to see that the AutCSP($\Ac$) is at least as difficult as the  CSP($\Gamma_{\Ac}$). Indeed, given any instance $I=(V, D, \mathcal C)$ of CSP($\Gamma_{\Ac}$) one transforms $I$ into the instance $I'=(V, D, \mathcal C')$
of the $\AutCSP(\mathcal A)$, where each $({\bs}, R)\in \mathcal C$ is mapped to $\ccl{\bs}$ from $\mathcal C'$. 

The $\AutCSP$ and $\AutCSP(\Ac)$ can be viewed as automata versions of uniform and non-uniform CSPs, respectively. Hence, we can pose the following research question: 
\begin{quote}{\em Is it true, similar to the classic CSP, that for every automaton $\Ac$, the  $\AutCSP(\Ac)$ is always  either decidable in polynomial time or  NP-complete?}    
\end{quote}
One theme in our approach is to determine the classes $\Kc$ of automata for which the problem $\AutCSP(\Ac)$ is in $P$ for all $\Ac\in \Kc$. 
%The way of identifying the class $\mathcal{G}$ is analogous to the way of identifying tractable constraint languages in non-uniform CSPs, e.g., via \textbf{polymorphisms} (see Section\ref{Sec:Pol_are_auto}).
For instance, in this paper we solved the question positively in the case when the automata $\Ac$ are given over the Boolean domain. 

\subsection{Examples and the single constraint case}

We now present several examples of $\AutCSP (\Ac)$:
\noindent 
\begin{example}\label{eg:boolean} In these examples $D=\{0,1\}$ is Boolean:
\begin{enumerate}
\item Consider $\CSP(\Gamma)$, where $\Gamma$ consists of finitely many relations of pairwise different arities. % is a finite language over $D$. 
This is an $\AutCSP$ as finite languages are regular\footnote{The relations should be 
of different arities as all relations in automatic constraint languages are of different arities. Otherwise, a bit of coding is required.}.

\item Let $\Ac_\mathrm{NAE}$ be an automaton, see  Figure~\ref{fig:nae-automaton}, recognizing  NAE $=\overline{0^* + 1^*}$, where NAE stands for not-all-equal.   
For all $k\geq 1$, the relation $R_k$ in $\Gamma_{\Ac_{NAE}}$ consists of all $k$-tuples $(d_1, \ldots, d_k)$ such that $d_i\neq d_j$ for some $i$ and $j$. The $\CSP(\mathrm{NAE})$ is NP-complete \cite{SchaeferDic}.
Since $\CSP(\mathrm{NAE})$ reduces in linear time to $\AutCSP(\Ac_\mathrm{NAE})$, the problem $\AutCSP(\Ac_\mathrm{NAE})$ is NP-hard.

\item Consider the automaton $\Ac_{\mathrm{ODD}}$ recognizing language $\mathrm{ODD} = \{\alpha\in (0+1)^*: \alpha \text{ has an odd number of 1s} \}$. In any instance of the $\AutCSP(\Ac_{\text{ODD}})$, each constraint has the  form $x_1+\dots +x_n=1$ mod $2$. Solving such constraints is in P.

\end{enumerate}
\end{example}

\begin{figure}
    \centering
   \includegraphics[width=0.7\linewidth]{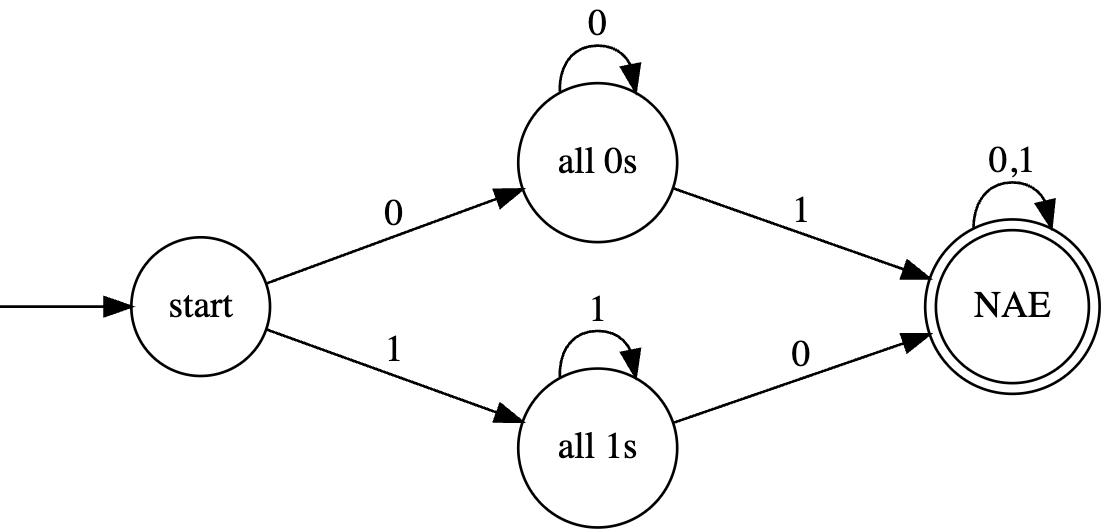}
    \caption{The automaton $\Ac_{\mathrm{NAE}}$}
    \label{fig:nae-automaton}
\end{figure}

Here are some basic facts about %the time complexity issues of the 
$\AutCSP$ and the $\AutCSP(\Ac)$.

%In order to show that the dichotomy conjecture makes sense in the setting of automatic CSPs,  we first need to show that  $\AutCSP$ is  in NP. The next observation does exactly that.

\begin{proposition} \label{Prop:$NP$-complete}
The following are true:
\begin{enumerate}
\item  The $\AutCSP$ is in NP. 
%Namely, the problem that, given an automatic instance $I$ from $\AutCSP$, decides if $I$ is satisfiable is an NP-problem. In fact, the problem is NP-complete.
\item The problem $\AutCSP(\Ac)$ is in NP for any automaton $\Ac$. Moreover, there exists a  automaton $\Ac'$ such that $AutCSP(\Ac')$ is $NP$-complete. 
\end{enumerate}
\end{proposition}
\begin{proof}
To see that the $\AutCSP$ is in NP, 
observe that for an $\AutCSP$ instance $I=(V,D,\Cc,\Ac)$ and an assignment $\phi:V \ar D$, each constraint $\ccl{\bs} \in \Cc$ can be checked by testing if the tuple $\phi(\bs)$ is accepted by $\Ac$. 
Likewise, for every fixed automaton $\Ac'$, the same argument shows that $\AutCSP(\Ac') \in NP$.

For NP-completeness, the problem $\AutCSP(\Ac_\mathrm{NAE})$ from Example \ref{eg:boolean} is NP-hard. So, set $\Ac'=\Ac_\mathrm{NAE}$.
Finally, as  $\AutCSP(\Ac_\mathrm{NAE})$ is essentially a subproblem of $\AutCSP$, the latter is NP, as well.
\end{proof}

We now show that $NP$-hardness persists even in a highly restricted case where $\AutCSP$ instances have only a single constraint. This already distinguishes the standard CSP from the $\AutCSP$. In the CSP setting, satisfying a single constraint is always in P.

\begin{theorem}\label{Thm:1-constraint}
Let $\Ac$ be an automaton over the Boolean domain $D = \{0,1\}$ such that $\mathcal L(\Ac) = \text{NAE}_3^*$, where $\text{NAE}_3 := D^3 \setminus \{000,111\}$. \ Then $AutCSP (\Ac)$ restricted to  single constraints is NP-complete.
\end{theorem}
\begin{proof}
We reduce from the well-known NP-complete problem $\CSP(\{\text{NAE}_3\})$ to AutCSP($\Ac)$.

Consider an instance $I = (V, D, \{C_k\}_{k=1}^m) \in \CSP(\{\text{NAE}_3\})$, where  we have $C_k = \text{NAE}_3(x_{i_k}, x_{j_k}, x_{l_k})$.
We now construct an instance $I' = (V, D, \{(\mathbf{s}, R_{3m})\}, \Ac)$ of $\AutCSP(\Ac)$ as follows. 

Let $\mathbf{s} := (x_{i_1}, x_{j_1}, x_{l_1}, x_{i_2}, x_{j_2}, x_{l_2}, \dots, x_{i_m}, x_{j_m}, x_{l_m}) $
be a tuple of length  $3m$ in which each block $(x_{i_k},x_{j_k},x_{l_k})$ of three variables corresponds to the constraint $C_k$ in $I$.  Let $R_{3m} := \mathcal L(\Ac) \cap D^{3m} = \text{NAE}_3^m.$

By construction, the assignment $\phi: V\rightarrow D$ satisfies all constraints in $I$ if and only if the tuple $\phi(\mathbf{s})$ belongs to $\text{NAE}_3^m$, i.e., it satisfies the single constraint $\ccl{\mathbf{s}, R_{3m}}$ in $I'$. Since both the reduction and membership testing for the fixed automaton $\Ac$ are polynomial-time, the problem is in NP and NP-hard, and thus NP-complete. 
\end{proof}

\subsection{Polymorphisms}

An important algebraic concept used in the study of the constraint satisfaction problem is the notion of polymorphism. 

\begin{definition}\label{dfn:polymorphisms}
Let $R$ be a relation of arity $k$ on $D$.
    A \textbf{polymorphism } of $R$ is an operation 
    $f:D^t \rightarrow D$ that satisfies the following condition. For all tuples $(x_{1,1}, \ldots, x_{1, k})\in R$, $\ldots$, $(x_{t,1}, \ldots, x_{t, k})\in R$ we have $
    (f(x_{1,1}, x_{2,1}, \ldots, x_{t,1}), \ldots, f(x_{1,k}, x_{2,k}, \ldots, x_{t,k}))\in R$.  The operation $f:D^t\to D$ is called a \textbf{polymorphism} of a language $\Gamma$ if $f$ is a polymorphism for every relation $R$ in $\Gamma$. 
\end{definition}

    For an automaton $\Ac$, a $t$-ary operation $f$ is said to be a \textbf{polymorphism of $\mathcal L(\Ac)$} if it is a polymorphism of $\Gamma_{\Ac}$, that is, if for any $x_1,\dots,x_t\in \mathcal L(\Ac)$ of the same length, $x_i=(x_{i,1},\dots,x_{i,k})$,
    $$
    (f(x_{1,1}, x_{2,1}, \ldots, x_{t,1}), \ldots, f(x_{1,k}, x_{2,k}, \ldots, x_{t,k}))\in R. 
    $$

A simple example of a polymorphism is the ternary majority operation on $D=\{0,1\}$ for the relation $R=\{(0,0,1), (0,1,0), (1,1,0)\}$. 
\begin{definition}\label{dfn:Pol-Inv}
 Let $\Gamma$ be a language and let $\Fc$ be a set of operations on $D$. We introduce the following two sets:
 \begin{itemize}
     \item $\Pol(\Gamma)=\{f\mid f$ is a polymorphism of $\Gamma\}$.
     \item $\Inv(\Fc)=\{R\mid$ every $f\in \Fc$ is a polymorphism of $R\}$. 
 \end{itemize}
 When $\Fc=\{g\}$ is a singleton, we write $\Inv(g)$ instead of $\Inv(\Fc)$.  Likewise, if $\Gamma=\{R\}$ is a singleton, we write $\Pol(R)$.  
 \end{definition}

The importance of polymorphisms can be seen from the following Dichotomy theorem for Boolean domains:
 
 %instead of $\Pol(\{R\})$.
% \anote{I suggest to use $\Gamma$ for constraint languages, as it is accepted in the CSP literature, and use $\Lc$ for languages recognized by automata.} \bakhnote{Xiaoyang}

%\smallskip

\par

\begin{untheorem}[Schaefer's Dichotomy \cite{SchaeferDic}] 
    Let $\Gamma$ be a (possibly infinite) set of relations over  $D = \{0,1\}$. Then   CSP($\Gamma$) is polynomial-time solvable iff one of the following is a polymorphism of \   $\Gamma$:
    \begin{itemize}
        \item the constant operations \  $\mathbf{ 0} : x \mapsto 0$ and \ $\mathbf{ 1} : x \mapsto 1$.
        \item the binary AND operation $\land : (x,y) \mapsto x \land y$.
        \item the binary OR operation $\lor : (x,y) \mapsto x \lor y$.
         \item the majority operation  
         $$Maj : (x,y,z) \mapsto ( x \land y) \lor (x \land z) \lor ( y \land z).$$
        \item the minority operation $Minor : (x,y,z) \mapsto x \oplus y \oplus z$, where $\oplus$ denotes addition modulo 2.        
    \end{itemize}
Otherwise, it is NP-complete.
\end{untheorem}
   
%% \textbf{ Theorem} [Schaefer's Dichotomy \cite{SchaeferDic}] \label{Thm:Boolean}\emph{ 
%%     Let $\Lc$ be a (possibly infinite) set of relations over the Boolean domain $D = \{0,1\}$. Then   CSP($\Lc$) is polynomial time solvable iff one of the following is a polymorphism of $\Lc$:
%%     \begin{itemize}
%%         \item the constant maps \  $\mathbf{ 0} : x \mapsto 0$ and \ $\mathbf{ 1} : x \mapsto 1$.
%%         \item the binary AND map $\land : (x,y) \mapsto x \land y$.
%%         \item the binary OR map $\lor : (x,y) \mapsto x \lor y$.
%%          \item the majority map  
%%          $$Maj : (x,y,z) \mapsto ( x \land y) \lor (x \land z) \lor ( y \land z).$$
%%         \item the minority map $Minor : (x,y,z) \mapsto x \oplus y \oplus z$, where $\oplus$ denotes addition modulo 2.        
%%     \end{itemize}
%% Otherwise it is NP-complete.}

One interesting class of operations on $D$ is the class of  Sigger's operations. These are the operations $f:D^4\rightarrow D$ such that 
for all $a,r,e\in D$ we have 
$f(r, a, r,e) = f(a, r,e,a)$. 
%The following theorem  shows that 
The Siggers polymorphisms are responsible for the $\CSP(\Gamma)$ problem to be in $P$:

\begin{untheorem}[\cite{DBLP:conf/focs/Bulatov17,DBLP:conf/focs/Zhuk17}]
    Let $\Gamma$ be a (possibly infinite) language on a finite domain. 
    Then the CSP($\Gamma$) is in $P$ if and only if \ 
    $\Gamma$ has a Sigger's polymorphism. Otherwise, it is NP-complete.
\end{untheorem}

\begin{remark}\label{rem:sizes}
An important point in these two theorems is that the sizes of CSP instances $I$ are based on  explicit representations  of \  $I$. 
\end{remark}
%% \textbf{ Remark 2}.  \emph{ The two theorems above bring an important point. In these theorems, the sizes of the CSP instances $I$ are based on  explicit representations  of $I$. Also, see \textbf{ Remark 1}. 
%, which are given by listing all the tuples of the relations that are present in the constraints of the instances. 
%% }

%% \smallskip

%We want to investigate the effect of polymorphisms on the complexity of $AutCSP(\Ac)$. For 

Now we provide examples of infinite automatic constraint languages that admit the operations from Schaefer's Dichotomy theorem 
above as polymorphisms. These are interesting examples because automatic representations of the CSPs in these languages are exponentially more succinct than their explicit representations. 

\begin{proposition} \label{Prop: poly-eg}
    %% Set $X=\{\mathbf{0}, \mathbf{1}, \land, \lor, Maj, Minor\}$. 
    %% For all $f\in X$ there exists an infinite automatic constraint language $\Gamma_{f}$ such that $f$ is the only polymorphism of $\Gamma_f$ from $X$. 
    Set $X=\{\mathbf{0}, \mathbf{1}, \land, \lor, Maj, Minor\}$. 
    For each $f\in X$ there exists an  automaton $\Ac_f$ over $D=\{0,1\}$ such that $\mathcal L(\Ac_f)$  is infinite and  $f$ is the only polymorphism of $\mathcal L(\Ac_f)$. 
\end{proposition}
\begin{proof}
For each $f\in X$, we provide a desired example $\Ac_f$. 

\noindent
\textbf{1}. 
Let $\Ac_{Maj}$ be such that $\Lc(\Ac_{Maj})=(001+010+110)^*$.  It is straightforward that $Maj$ preserves the set $\{001,010,110\}$. Since every word in $\Lc(\Ac_{Maj})$ is a concatenation of $3$-bit blocks from $\{001,010,110\}$, it immediately follows that  $Maj$ is a polymorphism of the language $\Lc(\Ac_{Maj})$.
Note that no other operation in $X$ preserves $\Lc(\Ac_{Maj})$, e.g., 
$Minor(001,010,110)=101 \notin \Lc(\Ac_{Maj})$,
$001 \land 010 =000 \notin \Lc(\Ac_{Maj})$, and 
$001 \lor 010 = 011 \notin \Lc(\Ac_{Maj})$.

\noindent
\textbf{2}. 
Let $\Ac_{Minor}$ be such that $\Lc(\Ac_{Minor})=\mathrm{ODD}$, the set of binary strings with an odd number of $1$s.
For any $x, y, z \in \Lc(\Ac_{Minor})\cap D^n$, set $w(i) := {Minor}(x(i), y(i), z(i))$. 
Then 
\[
\bigoplus_{i\in[n]}w(i) =
% \bigoplus_{i\in[n]}(x(i)\oplus y(i) \oplus z(i))= 
\bigoplus_{i\in[n]}x(i) \oplus \bigoplus_{i\in[n]}y(i) \oplus \bigoplus_{i\in[n]}z(i)=1,
\]
so $w \in \Lc(\Ac_{Minor})$, Thus $Minor$ is a polymorphism.
%The following shows that $\Lc(\Ac_{Minor})$ is not closed under the other operations:
Observe that 
$Maj (001,010,100) = 000 \notin \Lc(\Ac_{Minor})$,
$01 \land 10 =00 \notin \Lc(\Ac_{Minor})$,
$01 \lor 10 =11 \notin \Lc(\Ac_{Minor})$.

\noindent
\textbf{3}. 
%Let $\Ac_{\land}$ be such that 
Consider $\Lc(\Ac_{\land}) 
=1(000+001+010+011+100+101+110)^+$. 
This language consists of strings starting with $1$, followed by a sequence of 3-bit blocks excluding $111$. It is closed under $\land$, because applying $\land$ to any two blocks in the given set cannot result in the block $111$, which is excluded.
Observe that $Maj(011,101,110)=111$,
$Minor(001,010,100)=111$, and 
$001 \lor 110 = 111$.

%$The following shows that $\Lc(\Ac_{\land})$ is not closed under the other operations: 

\noindent
\textbf{4}. 
Similarly, let $\Ac_{\lor}$ be such that $\Lc(\Ac_{\lor}) 
=0(001+010+011+100+101+110+111)^+$.
Counterexamples:
%\begin{center}{
$Maj(001,010,100)=000$,
$Minor(011,101,110)=000$, and 
$001 \land 110 = 000$.
%}%\end{center}

\noindent
\textbf{5}. 
Let $\Ac_{\mathbf{0}}$ be such that $\Lc(\Ac_{\mathbf{0}})=(000+011+101)^{*}$. The operation 
$\mathbf{0}$ is a polymorphism of $\Lc(\Ac_{\mathbf{0}})$.
Other operations from $X$ fail $\Lc(\Ac_{\mathbf{0}})$. For instance: $011 \land 101 =001 \notin \Lc(\Ac_{\mathbf{0}})$, 
$011 \lor 101 = 111 \notin \Lc(\Ac_{\mathbf{0}})$,
$Maj(000,011,101)=001 \notin \Lc(\Ac_{\mathbf{0}})$, and we can refute $Minor$: 
$Minor(000,011,101)=110 \notin \Lc(\Ac_{\mathbf{0}})$. \ For the operation \textbf{1}, the language 
$\Lc(\Ac_{\mathbf{1}})=(111+100+010)^{*}$ is a desired one. 
\end{proof}

\section{Our contributions and discussions}\label{sec:contributions}

%Below we discuss contributions of this paper .

\textbf{1}.  The key notion introduced in this paper is the  Automatic Constraint Satisfaction Problem ($\AutCSP$) where constraint languages and instances are specified by finite automata (see Definition \ref{Dfn:AutCSP}). Unlike the classical non-uniform CSPs that are defined by finitely many relations, the $\AutCSP(\Ac)$ captures infinite yet finitely describable sets of relations, enabling concise representations of complex constraints. Studying the complexity of the $\AutCSP$s elucidates the interplay between classical CSPs, automata, and logic, sharpening the boundary between tractable and intractable regular constraints. We give two simple but interesting and illustrative examples of this interplay:
\begin{itemize}
\item The $\AutCSP$ and $\AutCSP(\Ac)$ are in NP by Proposition \ref{Prop:$NP$-complete}. 
    \item By Theorem \ref{Thm:1-constraint}, there exists an automaton $\Ac$ such that the 
$\AutCSP(\Ac)$ for instances with just one constraint is 
NP-complete. In contrast, solving the classical CSP instances with singleton constraints
is in polynomial time. %This contrasts the classical CSPs, where CSP instances with one constraint only are in $P$. 
\end{itemize}

\noindent
\textbf{ 2}. The AutCSPs instances are more \emph{succinct} than their classical counterparts.  We explain the importance of succinct representations  through  Schaefer's Dichotomy Theorem. By the theorem, the $\CSP(\Gamma)$ over the Boolean domain is in P if and only if  one of the operations $\mathbf{ 0}$, $\mathbf{ 1}$, $\land$, $\lor$, $Maj$, and $Minor$ is a polymorphism of $\Gamma$. The key is that the polynomial-time  algorithms in Schaefer's Theorem
run on the explicit representations of the $\CSP(\Gamma)$ instances.   Likewise, by  our  Theorem \ref{Thm:boolean-dichotomy}, we call it the automatic dichotomy theorem,   the $\AutCSP({\Ac})$ over Boolean domain is in P if and only if $\Gamma_{\Ac}$ admits one of the operations $\mathbf{ 0}$, $\mathbf{ 1}$, $\land$, $\lor$, $Maj$, and $Minor$ as a polymorphism. The key in our theorem is that our algorithms run on automatic instances rather than on explicitly represented instances. For example,  the languages $\mathcal L(\Ac_f)$ in Proposition \ref{Prop: poly-eg} have exponential growth. The sizes of the $\CSP(\Gamma_{\Ac_f})$ instances are 
exponentially larger than their automatic counterparts. Thus, our polynomial-time algorithms run on exponentially smaller CSP instances over these languages. 

%\smallskip

\noindent
\textbf{ 3}. The proof of the Dichotomy theorem for $\AutCSP$ over the Boolean domain  shows an interplay between polymorphisms, automata, 
and polynomial-time algorithms.  For instance, Theorem \ref{Thm:Booleam-AND-Polymporphism} shows that 
    the $\AutCSP(\Ac)$ is decidable in polynomial time if the automatic constraint language $\Gamma_{\Ac}$ admits $\land$ or $\lor$ as its polymorphism. The proof uses automata in extending partial solutions to the $\AutCSP$ instances. Likewise,  automata-theoretic reasoning is applied in Theorem \ref{thm:2minor_translation} showing that the same holds true if $\Gamma_{\Ac}$ admits the $Minor$ operation as a polymorphism.

%\anote{What about $Maj$. And what about NP-hardness? Should we state a general result saying that the classic CSP is always no harder than the same automatic one?} These results show how the standard CSP results can naturally be transformed into automata-theoretic settings. 

%\smallskip

\noindent
\textbf{ 4}. Any polymorphism $f:D^k \rightarrow D$ of a language $\Gamma_{\Ac}$ can be extended to a
partial operation  $f_{\omega}:D^*{^k}\rightarrow D^*$ as follows. Given $k$ words 
$u_1=d_{1,1}, \ldots, d_{1,t}$, $\ldots$, $u_k=x_{d,1}, \ldots, d_{k, t}$ over $D$ all having the same length $t$,
define the  word $f_{\omega}(u_1, \ldots, u_t)$ of length $t$ as follows:
$f_{\omega}(u_1, \ldots, x_t)=f(d_{1,1}, \ldots, d_{k,1}) \ldots f(d_{1,t}, \ldots, d_{k,t})$. \  By Corollary \ref{Cor:f-omega}, this partial operation $f_{\omega}$ is an automatic operation.  This simple yet important observation connects the algebraic theory of CSPs with automata and allows us to employ methods of the theory of automatic structures, e.g.,
\begin{itemize}
    \item Theorem \ref{Thm:deciding-polymorphism} proves that there is an algorithm that, given a finite set $\Fc$ of operations and an automaton $\Ac$ on $D$, decides if at least one of the operations from $\Fc$ is a polymorphism of the language $\Gamma_{\Ac}$. If $\Fc$ is fixed, then the algorithm runs in polynomial time on $\Ac$.
    Moreover, if there exists a fixed bound on the arity of the operations from $\Fc$, then the algorithm is polynomial on the sizes  of $\Fc$ and $\Ac$. 
\end{itemize}
This result, due to Corollary \ref{Cor:deciding-dichotomy}, implies that we can verify in cubic time if the problem $\AutCSP(\Ac)$, where $\Ac$ is an automaton over the Boolean domain, is in P or NP.

\textbf{ 5}. Although we  mostly aim at understanding $\AutCSP$ over the Boolean domain, we also provide results 
for arbitrary finite domains $D$ and use them for the Boolean case.  %We study  $\AutCSP$ for automatic constraint languages of several types. The analogous languages in the classic setting were among the first to be studied. 
Among them are automatic constraint languages of width 1, see \cite{barto2009constraint,16_jl&c_boundedwithhierarchy}, and languages with a near-unanimity polymorphism.  Our results are the following:
%% We address these results in the setting of automatic constrain languages:
\begin{itemize}
    \item Theorem \ref{thm:1-min autcsp} shows that any automatic CSP instance over $D$ can be transformed in polynomial time into an equivalent
    1-minimal automatic CSP instance. 
    \item Theorem \ref{Thm:automatic-withd1} then shows that the $\AutCSP({\Ac})$ is decidable in polynomial time if $\Gamma_{\Ac}$ has width $1$. 
    \item By Theorem \ref{Thm:Maj-Transformation}, if $\Gamma_{\Ac}$  has a majority (or a $k$-ary near-unanimity) polymorphism $g$, then we can transform, in P-time, instances $I\in \AutCSP(\Ac)$ to classical CSP instances over binary (or $k-1$-ary) relations closed under $g$. This shows that $\AutCSP(\Ac)$ is in P, see Corollaries~\ref{cor:majority} and ~\ref{cor:NU}. %in polynomial time. 
\end{itemize}

   % that Boolean-domain $AutCSP$s admitting a semilattice or minority polymorphism are solvable in polynomial time (Subsection \ref{booleansemi} and \ref{booleanmin}), and we further extend the semilattice result to finite domains (Section \ref{section:width_1}) , establishing tractability for the majority polymorphism as well (Section \ref{finitemaj}). These results together yield a dichotomy theorem for Boolean $AutCSP$s (Section \ref{boolean-dichotomy}) , showing that every such problem is either in $\mathbf{P}$ or $\mathbf{NP}$-complete, depending on the presence of a tractable automatic polymorphism.

%% Our results lay the groundwork for studying AutCSPs over infinite languages bridging CSP theory and automata theory by  offering new insights into CSP, polymorphisms, automatic structures, and the decidability limits of CSPs.
%% \anote{I'd remove this paragraph altogether. It doesn't mean anything} (\bakhnote{yes, lets remove it})

\section{Polymorphisms are automatic}\label{Sec:Pol_are_auto}

In this section, we investigate polymorphisms of automatic constraint languages. 
In particular, we design an algorithm that, given a finite family $\Fc$ of  operations and a finite automaton $\Ac$, 
%% an automatic constraint language $\Lc_{\Ac}$,  
decides if  at least one of the operations in $\Fc$ is a polymorphism of $\Gamma_{\Ac}$. We use Section \ref{S:Automatic-Relations} on automatic relations and structures.

Let $f: D^k\rightarrow D$ be an operation. Let us represent the operation $f$  as a $|D|^k \times (k+1)$  matrix such that each row of the matrix is of the form 
$(d_1, \ldots, d_k, f(d_1, \ldots, d_k))$. We can naturally \textbf{extend} $f$ to the $k$-ary partial operation $f_{\omega}$ on the domain $D^*$ of all strings over the domain $D$ as shown in Section~\ref{sec:contributions}-4.
%% follows. Given strings
%% $u_1=d_{1,1} \ldots d_{1,t}$, $u_2=d_{2,1} \ldots d_{2,t}$, $\ldots$, 
%% and  $u_k=d_{k,1}\ldots d_{k, t}$
%% all of these strings are of same length $t$, define
%% $$
%% f_{\omega}(u_1, \ldots, u_k)= (f(d_{1,1}, \ldots, d_{k,1}), \ldots, f(d_{1,t}, \ldots, d_{k,t})).
%% $$
%Thus, $f_{\omega}$ maps a $k$-tuple of strings of length $t$ to a string of length $t$.  

Note that $f_{\omega}$ is infinite while the original operation $f$ is finite. Now we prove the following automata-theoretic lemma: 
 
\begin{lemma}
For any $k$-ary operation $f$ on finite domain $D$, the partial operation $f_{\omega}$ 
on $D^*$ is an automatic (partial) operation.  
\end{lemma}

\begin{proof}
To prove that $f_{\omega}$ is automatic, we need to show that its graph is an automatic relation. Recall that the graph of $f_{\omega}$ is defined as the follows:
$(u_1, \ldots, u_k, u_{k+1}) \in Graph(f_{\omega})$ if and only if we have  $|u_1|=\ldots=|u_k|$ and $u_{k+1}=f_{\omega}(u_1, \ldots, u_k)$.

A deterministic finite automaton recognizing the graph of $f_{\omega}$ 
has two states: $s$ and $s'$. The alphabet of the automaton is $D^{k+1}$. 
The state $s$ is the start state and $s$ is the accepting state. The transitions from $s$ to $s$ are labeled by the rows of the $|D|^k \times (k+1)$  matrix representing the operation $f$. The labels of transitions from $s$ to $s'$ are  $(d_1,\ldots, d_k, d_{k+1})$ such that $f(d_1, \ldots, d_k)\neq d_{k+1}$. 
\end{proof}
%For instance, the automaton recognizing the graph of $f_{\omega}$, when $k=3$, is given in Figure \ref{fig:placeholder}. The lemma is proved.   \anote{I'm still not sure I understand the construction of the automata. Section~2.1 is kind of unrelated to that}  (\bakhnote{See the paragraph before the lemma and the figure. later, we can remove the figure if needed}) \anote{What is the alphabet of the automaton? It may be just your convention, but what if it is given $(a,b,c,d)$ and $d\ne f(a,b,c)$?}\begin{figure}[h] \centering  \includegraphics[width=0.5\linewidth]{pics/A_map.png}   \caption{The automaton for $f_{\omega}$, where $f: D^3 \to D$.}  \label{fig:placeholder} \end{figure}

Since polymorphisms are operations on the finite domain  $D$,
we get  the following corollary: 
\begin{corollary}[\textbf{Polymorphism are automatic}]\label{Cor:f-omega}
For all polymorphisms $f$ of \    
$\Gamma_{\Ac}$, the partial operation  $f_{\omega}$ is automatic. In particular, the partial operations \ $\mathbf{0}_{\omega}$, $\mathbf{1}_{\omega}$, $\land_{\omega}$, $\lor_{\omega}$, $Maj_{\omega}$, and $Minor_{\omega}$ over the Boolean domain are all automatic.
\qed 
\end{corollary}

%\anote{I'd suggest to remove this corollary, as being automatic has nothing to do with being a polymorphism}(\bakhnote{I re-wrote things. I think it should be clearer now. We identified $f$ with $\tilde f$ (which is now $f_{\omega}$ which confused you most likely. }) \anote{It's not quite that. The theorem says: all operations are automatic. Corollary 1: Operation A is automatic; Corollary 2: Operation B is automatic. True, but this is what the universal quantifier means}

%\begin{corollary} The partial maps \ $\mathbf{0}_{\omega}$, $\mathbf{1}_{\omega}$, $\land_{\omega}$, $\lor_{\omega}$, $Maj_{\omega}$, and $Minor_{\omega}$ over the Boolean domain are all automatic.\qed\end{corollary} \anote{I'd remove this one as well as too trivial,} (\bakhnote{hope it makes sense now.})

We can verify whether or not an operation on $D$ is a polymorphism of a given automatic constraint language $\Gamma_{\Ac}$. For instance, let us assume that $f:D^3\rightarrow D$ is a ternary operation. Then $f$ is a polymorphism of $\Gamma_{\Ac}$ 
if and only if the following statement is true:
 \begin{center}{
 $\forall x \forall y \forall z ((|x|=|y|=|z| \ \& \ x\in \mathcal L(\Ac) \ \& 
 \ y \in \mathcal L(\Ac) \ \& $ \\ $\ \& \ z\in \mathcal L(\Ac)) \rightarrow  f_{\omega}(x,y,z)\in \mathcal  L(\Ac)))$.}
 \end{center}
Note that this is a sentence in the first order predicate logic over the structure with domain $D^*$, and the relations on the domain $D^*$:  $\{(x,y) \mid |x|=|y|\}$, $\{x\mid x\in \mathcal L(\Ac)\}$, $\{(x,y,z,u) \mid f_{\omega}(x,y,z)=u\}$. By the Decidability Theorem \ref{DecidabilityThm}, the truth value of the sentence above can  be verified. Yet, we have a stronger statement:

\begin{theorem}\label{Thm:deciding-polymorphism}
The following are true:
\begin{enumerate}
\item Given a $k$-ary operation $f : D^k \ar D$ and a finite automaton $\Ac$, it takes $O(|\Ac|^{k+1})$ time to decide whether the constraint language $\Gamma_{\Ac}$ has $f$ as its polymorphism
\item Fix an integer $b\geq 0$. There exists a polynomial-time algorithm that, given a finite set $\Fc$ of operations on $D$ whose arities are bounded by $b$ and an automaton $\Ac$, decides if any of the operations is a polymorphism of the language $\Gamma_{\Ac}$. The algorithm runs in time $O(|\Fc||\Ac|^{b+1})$. 
\end{enumerate}
\end{theorem}
\begin{proof}
We prove the first part. The second part of the theorem follows from the first part. Consider the automaton $\Ac=(S, q_0, \Delta, F)$.
We show that the problem of deciding whether $f$ is a polymorphism of the constraint language $\Gamma_\Ac$ can be reduced to the emptiness problem for finite automata. Indeed, define the following finte automaton $\Ac_f = (S_f, q_f, \Delta_f, F_f)$, where 
\begin{itemize}
\item $S_f = S^{k+1}$,
\item The initial state  $q_f$ is the tuple $(q_0, \ldots, q_0)$, 
\item $F_f=F^k \times (S\setminus F)$, and
\item The transition diagram $\Delta_f$ = 
$$ \bigcup_{
       \substack{
       %\xynote{s_i \in S, d_i\in D} \\
       d_{k+1} = f(d_1, \ldots, d_k) \\
       t_i = \Delta(s_i, d_i), \ i \in [k+1] \\
       }
       } \{((s_1,\ldots, s_{k+1}), (d_1, \ldots, d_{k+1}), (t_1,\ldots,t_{k+1}))\}
$$
\end{itemize}
Note that  the $s_i$'s and $d_i$'s come from the transition function  
$\Delta: S\times D\rightarrow S$ of the automaton $\Ac$.
It is not hard to see that an $k+1$-tuple $(x_1, \ldots, x_k, x_{k+1})$ is accepted by $\Ac_f$ if and only if we have 
  $x_{k+1}\notin \mathcal L(\Ac)$, $x_i \in \mathcal L(\Ac)$ for each $1\le i \le k$, and $x_{k+1} = f(x_1,\ldots, x_k)$. Therefore, any tuple accepted by $\Ac_f$ would  refute that $\mathcal L(\Ac)$ has $f$ as its polymorphism. 
%% The size of $S_f = S^{k+1}$ shows that it takes $O(S^{k+1})$ time to found such tuple. 
Thus, $f$ is a polymorphism of the language$\Gamma_{\Ac}$ if and only if $\Ac_f$ accepts no string.
  The construction of automaton $\Ac_f$ takes $O(|\Ac|^{k+1})$ time.
Since the emptiness problem for finite automata can be decided in linear time relative to the size of the automata, the theorem is thus proved.
\end{proof}

%One theme in this paper is to identify operations $f$ such that the problem $\AutCSP(\Ac)$ is solvable in polynomial time, where the language $\Lc_{\Ac}$ admits  $f$ as a polymorphism. 

Theorem \ref{Thm:deciding-polymorphism} has the following corollaries. The first concerns automatic constraint languages over the Boolean domain.
 
\begin{corollary}\label{Cor-P-decision}
  There is a polynomial-time algorithm that, given a finite automaton $\Ac$ over the domain $\{0,1\}$, decides if
  any of the operations $\mathbf{0}$, $\mathbf{1}$, $\land$, $\lor$, $Maj$, and $Minor$ is a polymorphism of the automatic constraint language $\Gamma_{\Ac}$. The algorithm runs in time $O(|\Ac |^4)$. \qed 
 \end{corollary}

 The second is on Sigger's operations. The proof is in Appendix~\ref{S:Pol-Check}.
 %In this corollary an important assumption is that the domain $D$ is fixed.
 
 \begin{corollary} \label{Cor:Pol-check}
Let $\Ac$ be a finite automaton over $D$. Then:
\begin{enumerate}
    \item It is decidable in polynomial time whether or not $\Gamma_{\Ac}$ admits a Sigger's operation as polymorphism.
    \item %Moreover, 
    If no Sigger's operation is a polymorphism of $\Gamma_{\Ac}$, then we can find in time $O(|\Ac|^5)$ a finite sub-language $\Gamma \subseteq \Gamma_{\Ac}$ such that the $CSP(\Gamma)$, and hence $AutCSP(\Ac)$, are NP-complete. 
    %% Hence, if none of the Sigger's function is a polymorphism of $\Lc_{\Ac}$, then we can find a sub-language $\Lc \subseteq \Lc_{\Ac}$ witnessing 
    %% NP-completeness of 
    % $\CSP(\Lc_{\Ac})$.
    %% $AutCSP(\Ac)$.
    \qed 
\end{enumerate}  
 \end{corollary}

 Theorem~\ref{Thm:deciding-polymorphism} and its corollaries motivate  us to give the following definition that we will  use in the rest of  this paper.

\begin{definition}[\textbf{ The class $\Kc_f$}]
Let $f : D^k \to D$ be an operation.  Define the
following class of automata:
\[
\Kc_f \;=\;
\{\Ac \mid \text{The language $\Gamma_{\Ac}$ has $f$ as its polymorphism} \}.
\]
\end{definition}

\section{Automatic CSPs over Boolean domain}\label{S:BooleanDomain}

In this section, we investigate two cases of 
the $\AutCSP$ over the Boolean domain $D=\{0,1\}$.    
 The first case concerns the $\AutCSP$ for instances $I=(V, D, \Cc, \Ac)$, 
 where $\Ac \in \Kc_{\land }$ or $\Ac \in \Kc_{\lor}$.
The second case focuses on  the $\AutCSP$ for instances $I=(V,D,\Cc, \Ac)$,  
 where $\Ac \in \Kc_{Minor }$.  By Theorem \ref{Thm:deciding-polymorphism}, given a finite automaton $\Ac$, we can decide in polynomial time if $\Ac \in \Kc_{f}$ for each  $f \in \{\land, \lor, Minor\}$.   

\subsection{The case of $\land$ and $\vee$ polymorphisms} \label{booleansemi}

We investigate $\AutCSP$ instances $I=(V, D, \Cc, \Ac)$ where $D=\{0,1\}$ and $\Ac \in \Kc_{\land}$. The case of $\vee$ is analogous.  We show how the automaton $\Ac$ can be used to derive a polynomial-time decision algorithm for the $\AutCSP$ instance $I$. 
We introduce the core combinatorial structure and  the  technical  ideas in a clean way.
The next section extends these ideas to the class of automatic constraint languages of width 1. 
Since tuples are functions, we will use a partial function to describe a tuple with some "pattern": 

\begin{definition}
    A partial function $\tau:[n]\rightharpoonup D$ is  a \textbf{ pattern}. 
    A function $\psi:[n]\rightarrow D$ \textbf{extends} $\tau$ with respect to an $n$-ary relation $R_n $ on $D$ if $\psi(i)=\tau(i)$ for all $i\in dom(\tau)$ and $(\psi (1),\dots, \psi(n))\in R_n$.
    A pattern $\tau$ is \textbf{extendable} if it has at least one extension.
\end{definition}\label{dfn:pattern}

\begin{lemma}\label{lem:extend}
Let $\Ac=(S,q_0,\Delta,F)$ be a finite automaton over the alphabet $D$, $n\in\mathbb N$, and $\tau:[n]\rightharpoonup D$ be a pattern. Then the extendability of $\tau$ can be checked in time $O(n\,|\Ac|)$.
\end{lemma}
\begin{proof}
The extendability can be tested via a BFS on $\Ac$ to depth $n$:  at each position $i$, follow $\Delta(s_{i-1},\tau(i))$ if $i\in\mathrm{dom}(\tau)$; otherwise, any outgoing transition may be used.  %Both the test and construction run in $O(n\,|\Ac|)$ time.
\end{proof}

Patterns represent partial solutions to a given instance. For a given pattern our task is to extend the pattern to a full solution. % of the instance. 

%The main theorem of this section is the following.

\begin{theorem}\label{Thm:Booleam-AND-Polymporphism}
There exists a polynomial-time algorithm that, given
an $\AutCSP$ instance $I =(V,D,\Cc,\Ac)$ where $\Ac\in \Kc_{\wedge}$, decides if the instance $I$ has a solution. The algorithm runs in time $O(|V|\,|\Cc|^2\,|\Ac|)$ where $|\Cc| = \sum_{\ccl{\bs, R}\in \Cc} |\bs|$.
\end{theorem}
Our proof %of Theorem~\ref{Thm:Booleam-AND-Polymporphism} 
proceeds in three steps: 
First, we compute the minimal elements of $\land$-closed automatic relations (Lemma~\ref{lemma:meet_downward}). Second, we solve single-constraint instances (Lemma~\ref{lemma:single_constraint}). Finally, we extend the approach to %handle 
instances with multiple constraints.

Let $\Ac \in \Kc_{\land}$ be a finite automaton over  $D=\{0,1\}$.
Recall that $R_n=\mathcal L(\Ac)\cap D^n$. 
Since $R_n$ is $\land$-closed, we introduce a partial order on $R_n$ by declaring that $r_1\le r_2$ if and only if $r_1\land r_2=r_1$.

In particular, $R_n$ has a unique minimal element $\bigwedge_{r\in R_n}r$, i.e., the tuple in $R_n$ with the maximum number of zeros.
Explicitly computing $\bigwedge_{r\in R_n}r$ may take exponential time. To address this, we introduce the concept of minimal extension, which allows us to find this minimal element efficiently.

Formally, given an extendable pattern $\tau:[n] \rightharpoonup D$, a tuple $r \in R_n$ is called a \textbf{minimal extension} of $\tau$ if $r$ extends $\tau$ and, for every other extension $r'\in R_n$, we have $r\leq r'$. 

%\anote{Should we keep the same alphabet for patterns? Say, Greek for all of them? Or do you want to distinguish patterns and extensions?}

\begin{lemma}\label{lemma:meet_downward}
There exists an algorithm that, given a finite automaton $\Ac \in \Kc_{\land}$  and a pattern $\tau : [n] \partialar D$, decides whether $\tau$ is extendable and computes the minimal extension if such an extension exists. The algorithm runs in time $O(n^2|\Ac|)$.  
\end{lemma}

\begin{proof}
To compute the minimal extension of $\tau$, we start from any valid extension $v_0$ that can be found using Lemma~\ref{lem:extend} and iterate through positions $i=1,\dots,n$.  
Whenever $i\notin\mathrm{dom}(\tau)$ and $v_{i-1}(i)=1$, we temporarily fix $\tau(i)=0$ and test extendability again.  
If extendable, we denote by $m_i$ the new extension and update $v_{i}=v_{i-1}\land m_i$; otherwise, we keep $v_{i}=v_{i-1}$.

The resulting sequence $\{v_i\}_{i=0}^{n}$ is a non-increasing sequence  with respect to  $\leq$. Hence, the sequence stabilizes at the minimal element consistent with $\tau$.  
Thus, the final tuple $v_{n}$ is the minimal extension.  
At most $n$ extendability tests are performed, each requiring $O(n|\Ac|)$ time, so the overall complexity is $O(n^2|\Ac|)$.  
In particular, when $\tau = \emptyset$ %inLemma~\ref{lemma:meet_downward} yields 
we get the minimal element of $R_n$.
\end{proof}

We next compute minimal extensions for one constraint: 

\begin{lemma}\label{lemma:single_constraint}
There is an algorithm running in time $O(|V|\,n^2\,|\Ac|)$ that, given 
an automaton $\Ac \in \Kc_{\land}$, 
a constraint $C = \langle \bs, R_n \rangle$, and
a pattern 
$\tau : [n] \rightharpoonup D$ ,
decides if $\tau$ can be extended to a tuple satisfying $C$, 
and, if so, computes the minimal satisfying extension of $\tau$.
\end{lemma}

\begin{proof}
Lemma~\ref{lemma:meet_downward} allows one to find a minimal extension of $\tau$ in $R_n$. However, this extension may not satisfy the constraint $\langle \bs, R_n \rangle$, as some of the entries of $\bs$ may be equal. We describe a simple iterative procedure and argue correctness and complexity. 

\begin{enumerate}
  \item Using Lemma~\ref{lemma:meet_downward}, compute the minimal extension $r \in R_n$ of the current $\tau$ since $R_n$ is closed under $\wedge$. If none exists, return "no extension".
  \item If the relation $\{(s_i, r_i)\}_{i=1}^n$ defines a consistent assignment to variables (i.e., it is a function $s_i\mapsto r_i$), then $r$ is a satisfying tuple for $C$; output the induced extension,  and stop.
  \item Otherwise, there is a variable $v\in V$ appearing at positions $i, j$ with $r_i\neq r_j$. By minimality of $r$, any satisfying tuple must assign $v=1$. Fix all occurrences of $v$ in the constraint to $1$ by extending $\tau$ accordingly, and go back to Step~(1).
\end{enumerate}
Each iteration either finds a minimal satisfying extension or fixes a variable that must be assigned $1$, so the algorithm is correct. Each iteration uses Lemma~\ref{lemma:meet_downward} once,  and at most $|V|$ iterations occur, giving total time $O(|V|\,n^2\,|\Ac|)$. Note that we used the polymorphism $\wedge$ in the algorithm above through  Lemma \ref{lemma:meet_downward}. 
\end{proof}

Similarly, setting $\tau = \emptyset$ in Lemma~\ref{lemma:single_constraint} yields a solution to the single-constraint $\AutCSP$ instances.  

\begin{proof}[Proof of Theorem~\ref{Thm:Booleam-AND-Polymporphism}]
We extend the approach above to handle instances with multiple constraints 
$\Cc=\{C_1,\dots,C_m\}$, where each $C_i=\langle \bs_i, R_{k_i}\rangle$, 
which leads to a polynomial-time algorithm as stated in  Theorem~\ref{Thm:Booleam-AND-Polymporphism}.

\begin{enumerate}
    \item Initialize every partial assignment $\tau_i=\emptyset$.
    \item For each $C_i\in\Cc$, compute its minimal consistent tuple $e_i$ using Lemma~\ref{lemma:single_constraint} and let $\varphi_i$ be the corresponding function. If $\varphi=\bigcup_i \varphi_i$ is consistent, output $\varphi$ as a solution.
    \item If $\varphi$ is inconsistent, find the first variable $v$ witnessing inconsistency.  
    For each $C_i$ containing $v$, extend $\tau_i$ by fixing all coordinates of $v$ in $\bs_i$ to $1$.
    \item For each $C_i$, recompute the consistent minimal $e_i'$ under the new $\tau_i$.  
    If some $C_i$ becomes unsatisfiable, output "no solution"; otherwise, return to step~(2).
\end{enumerate}

Each iteration fixes at least one variable, so the number of iterations is at most $|V|$.  
Within each iteration, Lemma~\ref{lemma:meet_downward} is called for all constraints, taking time $O(|\bs_i|^2|\Ac|)$ per constraint.  
Therefore, the total runtime is $O(|V|\,|\Cc|^2\,|\Ac|)$. 
%where $|\Cc| = \sum_{\ccl{\bs, R}\in \Cc} |\bs| $.
\end{proof}

\subsection{The case of  minority polymorphism}\label{SubS:Minor}

We assume that the operation $Minor$ is a polymorphism of the automatic constraint language $\Gamma_{\Ac}$  over the Boolean domain. 
For this section, we expand the group $G_D = (\{0,1\}, \oplus)$ to the two-element field $GF(2) \cong (\{0,1\}, \oplus, \land)$. 

We prove that any $n$-ary relation $R\in \Gamma_\Ac$ can be identified with an equivalent system of linear equations $Mx = b$, that is, $r\in R$ if and only if $M r^{\intercal} = b$. While this correspondence is well understood, we show how to extract the matrix $M$ from the automaton $\Ac$ only.

It is well-known that $R\subseteq D^n$ admits $Minor$ as its polymorphism if and only if $R$ is a coset of some subgroup $S\le G_D^n$.  Moreover, $G^n_D$ is a vector space over the field $GF(2)$. Hence, $S$ is a subspace of the space $G^n_D$.
%and thus $R$ is an affine subspace. 
The subspace $S$ is a kernel of some matrix $M$, that is  $M \bs^{\intercal}  = 0$ iff $\bs \in S$. In such a representation, $R$ is equivalent to the set of solutions to the linear equation $M x = b$ for some vector $b$.  

\begin{theorem}\label{thm:2minor_translation}
   There exists an algorithm that, given an automaton $\Ac \in  \Kc_{Minor}$ over the Boolean domain $D$ and an integer $n$,   constructs in time $O(n^3|\Ac|)$ a system of linear equations $M x = b$ such that $r\in R_n$ if and only if $M r^{\intercal} = b$.
\end{theorem}
\begin{proof}
    We fix a tuple $r\in R$ and set $S = R - r$, a linear subspace of $D^n$ of dimension $m$. We construct a basis $\Bc$ of $ S$ in $m$ steps via a strictly increasing chain of bases $\Bc_1 \subsetneq \Bc_2\subsetneq \ldots \subsetneq \Bc_m = \Bc$.
    Once $\Bc$ is constructed, choose a maximal set of linearly independent row vectors  
    $\{\alpha \in D^{1\times n}: \forall b\in \Bc (\alpha\cdot b^{\intercal} = 0)\}$,   and let $M$ be the matrix whose rows are these vectors. Then $\ker M = \mathrm{span}\,\Bc = S$.

    Let $\Bc_0=\emptyset$. At step $t(t\ge 1)$, given $\Bc_{t-1} = \{b_k: 1\le k \le t-1\}$, we seek $b_t \in S $ such that $\Bc_{t} = \Bc_{t-1}\cup\{b_t\}$ is a basis of size $t$. 
    
To that end, we normalize the matrix $\Bc_{t-1}$ to a row canonical form. Let $L_{t-1}\subseteq [n]$ be the 
        columns containing the leading $1$s. Then a valid candidate $b_t$ exists iff there exists a non-zero $b_t\in S$ such that $b_t(k) = 0$ for all $k\in L_{t-1}$.
        Given $L_{t-1}$, we can detect such a candidate by testing the extendability of the following patterns. For each  column  $p\in [n]\setminus L_{t-1}$, define a pattern $\phi_{p,L_{t-1}}:[n]\partialar \{0,1\}$:  
         \[
        \phi_{p,L_{t-1}}(k) = \begin{cases}
            r(k) \oplus 1 , k=p \\ 
            r(k) , k\in L_{t-1} \\
            \text{undefined}, \text{ otherwise}
        \end{cases}
        \]
 There is an $x\in R$ extending $\phi_{p,L_{t-1}}$ iff $b = x - r\in S$ satisfies $b(k) = 0$ for all $k\in L_{t-1}$ and $b(p) = 1$, in which case we set $b_t = b$.

The process is summarized in Algorithm~\ref{Algo:basis_construct}. 
We now analyze the running time of this algorithm. The initialization takes time $O(n|\Ac|)$ to find a fixed vector $r$. The algorithm would iterate at most $n$ times. During each iteration, we test the extendability of each pattern $\phi_{p,L_{t-1}}$ on $\Ac$ and it takes time $O(n^2|\Ac|)$ in total. Therefore, the kernel $R -r =S$ can be constructed in time $O(n^3|\Ac|)$ for any relation $R\subseteq D^n$.
 This concludes the proof of the theorem.
 \end{proof}

 {\renewcommand{\figurename}{Algorithm}
    \begin{figure}
        \centering
        \begin{tabular}{l}
            \textbf{Input}: An automaton $\Ac$, an arity $n$\\ 
            \textbf{Output}: $(r, \Bc)$ such that $r + S = R_n$ with $S$ spanned by $\Bc$ \\
            \algoquad $\Bc \gets \emptyset$; \   \   \   \  \  \   \   \   \  \  
           {\em  Basis of kernel}  $S$\\ 
            \algoquad $L_0 \gets \emptyset$; \   \   \   \  \  \   \   \   \  \    {\em Columns of leading "$1$"s of basis} $\Bc$ \\
            \algoquad Find an accepting tuple $r\in \mathcal L(\Ac)$ of length $n$;\\
            \algoquad \textbf{For each} $t\in [n]$ \textbf{do}\\
            \algoquad\algoquad$\text{Changed} \gets \bot$; \   \   \   \  \  \   \   \   \  \ {\em A flag if $\Bc$ was extended} \\
            \algoquad\algoquad\textbf{For each} $p\in [n]\setminus L_{t-1}$ \textbf{do} \\
            \algoquad\algoquad\algoquad 
            \textbf{If} $\phi_{p,L_{t-1}}$ is extendable to $b$ \textbf{then} \\ 
            \algoquad\algoquad\algoquad\algoquad 
            Changed $\gets \top$;  
            %\\ \algoquad\algoquad\algoquad\algoquad 
            $L_{t} \gets L_{t-1}\cup\{p\}$; \\
            \algoquad\algoquad\algoquad\algoquad 
            \textbf{For each} $v\in \Bc$ \textbf{do} \\
            \algoquad\algoquad\algoquad\algoquad\algoquad
            \textbf{If} $v(p) = 1$ \textbf{then} $v \gets v + b$;\\
            \algoquad\algoquad\algoquad\algoquad \textbf{Done} \\
            \algoquad\algoquad\algoquad\algoquad
            $\Bc \gets \Bc \cup \{b\}$;\\
            \algoquad\algoquad\algoquad\algoquad
            Normalize $\Bc$ into row canonical form; \\
            \algoquad\algoquad\algoquad\algoquad
            \textbf{Break};\\
            \algoquad\algoquad \textbf{Done} \\

            \algoquad\algoquad \textbf{If} Changed $== \bot$ \textbf{then} 
            \textbf{Break} ; \\
            \algoquad \textbf{Done} \\
            \algoquad \textbf{Return} $(r,\Bc)$;
        \end{tabular}
        \caption{Basis Construction}
        \label{Algo:basis_construct}
    \end{figure}}
Theorem~\ref{thm:2minor_translation} shows that any automatic relation with the $Minor$ polymorphism can be translated into an equivalent system of linear equations over $GF(2)$.  
This leads directly to a polynomial-time algorithm for solving $\AutCSP$ by reducing it to solving a single global system of linear equations.

\begin{theorem}\label{thm:AutCSP-solver}
There is a polynomial-time algorithm that, given
an $\AutCSP$ instance $I =(V,D,\Cc,\Ac)$ where $\Ac\in \Kc_{Minor}$, decides if the instance $I$ has a solution. The algorithm runs in time $O(|\Cc|^3|\Ac|+|V|^3)$ where $|\Cc| = \sum_{\ccl{\bs, R}\in \Cc} |\bs|$.
\end{theorem}

\begin{proof}
For each constraint $\ccl{(v_{j_1}, \ldots, v_{j_{n_i}}), R_{n_i}}$,
by Theorem~\ref{thm:2minor_translation},  
we can build an equivalent system of linear equations $M_i x_{\bs_i} = b_i$,
where $x_{\bs_i} = (x_{v_{j_1}}, \ldots, x_{v_{j_{n_i}}})^{\intercal}$. The matrix $M_i$ is over the field $GF(2)$, and every equation in $M_i x_{\bs_i} = b_i$ cannot contain a variable more than once. Therefore, the length of equation is bounded by $|V|$.

All constraints $C_1,\dots,C_m$ are combined into a global system $M x = b$ of linear equations, where $x = (x_v)_{v \in V}$ is the vector of all variables, and $M$ and $b$ are obtained by vertically concatenating all local matrices $M_i$ and vectors $b_i$, with each $M_i$ extended with zeros in columns corresponding to variables not in $\bs_i$.

Thus, a function $\phi: V \to D$ satisfies  $I$ iff the corresponding Boolean vector $x = (\phi(v))_{v\in V}$ 
satisfies the system $M x = b$. 
Each local translation takes $O(n_i^3|\Ac|)$ time, and the system is solved by Gaussian elimination in $O(|\Cc||V|^2)$ where $|\Cc| = n_1+\dots+n_m$.  Hence, the overall procedure runs in time $O(|\Cc|^3|\Ac|+|\Cc||V|^2)$. 
\end{proof}

\textbf{ Epilogue: finite fields case}. The Boolean domain case for minority polymorphism can easily be extended to  finite fields. % byessentially following the Boolean case. 
Below we state the results with proofs in Section \ref{Appendix: GF(q) extraction} of the Appendix.

Over the field $GF(q)$, consider the operation $f_q(x, y, z) = x - y + z$. When $q=2$, the operation $f_q$ is the minority operation. It is well-known that every relation $R$ in $\Inv(f_q)$ is equivalent to a finite system of linear equations over $GF(q)$ \cite{Bulatov2006SimpleMaltsev}. %, i.e., $R$ is an affine subspace (a coset of a linear subspace). 
Thus, the classical $\CSP(\Inv(f_q))$ is decidable via Gaussian elimination. \ When $R\in \Gamma_{\Ac}$ for $\Ac \in \Kc_{f_q}$, we must recover its linear description from $\Ac$:

%The following lemma shows that. 
%The map $f_q$ is the Maltsev operation on the additive group reduct of $GF(q)$.
%Now we outline the  set-up for the finite fields and state the results with proofs in Section \ref{Appendix: GF(q) extraction} of the Appendix. The proofs essentially follow the Boolean case. Over the field $GF(q)$, consider the operation $f_q(x, y, z) = x - y + z$. When $q=2$, the operation $f_q$ is the minority operation. 
%The map $f_q$ is the Maltsev operation on the additive group reduct of $GF(q)$. 

%that this extraction is possible in polynomial time.

\begin{lemma}\label{lemma:finite_field_extraction}
   There exists an algorithm, which given an arity $n$ and an automaton $\Ac\in \Kc_{f_q}$, constructs in time $O(n^3|\Ac|)$a system of linear equation $Mx = b$ such that $r\in R_n$ if and only if $Mr^\intercal = b$. \qed 
\end{lemma}

Hence, we have the following theorem:
\begin{theorem}\label{thm:AutCSP-finite-field}
There is a polynomial-time algorithm that, given an $\AutCSP$ instance 
$I=(V,D_q,\Cc,\Ac)$ where $D_q$ is the underlying domain of the finite field $GF(q)$ and $\Ac\in \Kc_{f_q}$,
decides whether $I$ has a solution.
The algorithm runs in time 
$O(|\Cc|^3|\Ac| + |\Cc||V|^2)$. \qed 
%, where  $|\Cc| = \sum_{\ccl{\bs, R}\in \Cc} |\bs|$.
\end{theorem}

\section{Width 1 Automatic Constraint Languages}\label{section:width_1}

In this section, we focus on the CSP over automatic constraint languages of width $1$. So, let us fix a finite domain $D$, a finite automaton $\Ac = (S, q_0, \Delta, F)$ over $D$, and the language $\Gamma_{\Ac}$. 
It is known that languages closed under a semilattice operation have width $1$. 
%In particular, it is true for the Boolean operations $\vee$ and $\wedge$. 
%In this section, let $\land$ denote a semilattice operation on $D$.

\subsection{Pattern-Driven Search}

An important ingredient in finding a solution to a given constraint is a process that refines a given partial solution. We need to recast this process in automata-theoretic setting.

The notion of a pattern in Definition \ref{dfn:pattern} is single-valued and partial, which is sufficient in the Boolean setting: each position either receives a fixed value from $\{0,1\}$ or remains undefined, and our algorithm in Subsection \ref{booleansemi} extends such partial assignments until a full solution is obtained.
However, this framework breaks down for width-$1$ constraint languages over arbitrary finite domains.
In these settings positions may need to carry sets of possible values rather than a single fixed value. In addition, the algorithm in this section starts with the full domain $D$ at every position and enforces consistency by shrinking these sets rather than extending a single-valued partial assignment as in the Boolean case.
This key difference motivates the notion of a set-valued pattern.

\begin{definition}
Given a domain $D$, a \textbf{ set-valued pattern} is a function  $\phi:[n] \ar \mathcal P(D)$.
\end{definition}

%Intuitive explanation of this definition is the following.  In the Boolean case,  patterns were single valued partial functions  $\phi:[n] \ar D$. Our goal was to find extensions of patterns that satisfy a given Instance. In the Boolean domain case, this idea worked. It turns out in the case of constraint languages of width $1$, instead of single valued patterns we need set-valued patterns to find solutions to constraints.  

%One can view  a pattern $\phi$ as a partial solution to a constraint.  If $\phi (i)$ is defined, then the values of the variable at position  $i$ belong to the set $\phi(i)$. Otherwise, $\phi(i)$ is undefined, and weseek to extend the pattern by searching for a set-value of the variable at position $i$.  A partial function $\psi:[n] \rightarrow \mathcal P(D)$ \textbf{ extends} $\phi$ if for all $i\in dom (\phi)$ we have $\phi(i)=\psi(i)$. The extension $\psi$ is \textbf{ proper} if there exists a $j\in [n]$ such that $\phi(j)$ is undefined but $\psi(j)$ is defined. 

Given a set-valued pattern $\phi$, we want to find a tuple $t$ such that 
$t\in R_n \cap \phi(1) \times \ldots \times \phi(n)$.  Finding such a tuple  
can be exponential on $n$. The next lemma shows that exponential search is not necessary for automatic constraint languages. A proof is  in Appendix~\ref{S:procedure:search_pattern}.

\begin{lemma}\label{procedure:search_pattern}
    There exists an algorithm that, for a given  pattern $\phi: [n] \ar \Pc(D)$ and an automaton $\Ac$, decides if there exists a tuple $t\in \mathcal L(\Ac)\cap \prod_{i=1}^n\phi(i)$ and, if so, constructs such a tuple $t$. The algorithm runs in time  $O(n|\Ac|)$. 
    \qed 
\end{lemma}

\subsection{1-minimality in classical setting}

We introduce the \textbf{1-minimality} concept and then define constraint languages of \textbf{width 1}. For an instance \(I\) of \(\CSP(\Gamma_H)\), the  1-minimality concept mimics the termination of the unit propagation process in \textsc{Horn-3-SAT} as presented below.

%The notion of width 1 is closely tied to 1-minimality in constraint satisfaction problems. Unit propagation in \textsc{Horn-3-SAT} iteratively derives inferences from known variables until a fixpoint is reached, and the 1-minimality algorithm generalizes this propagation scheme.
%% \anote{I don't understand what you say about Horn-3-SAT} (\bakhnote{We rephrased it by saying we explain 1-minimality through the unit propagation for Horn-3-SAT. It might help automata people to follow the idea. }).

\begin{example}[\textsc{Horn-3-SAT}]
The language for \textsc{Horn-3-SAT} is $\Gamma_{H} = \{S_{110}, S_{111}, C_1\}$, where $S_{xyz} = \{0,1\}^3\setminus\{(x,y,z)\}$ and $C_1 = \{1\}$.  
%The unit propagation is the procedure which uses known values of variables to determine possible values of other variables. 
%For instance, consider the constraint  \[  \Cc = \{C_1(x), S_{111}(y,y,y), S_{110}(x,z,y)\} \] $\Cc$ forces that $x = 1,y=0$ and thus determines $z = 0$ because $S_{110}$ has exactly one tuple $t = (1,0,0)$  satisfying $\pi_1 t = 1, \pi_3 t = 0$. In other words, values $x=1,y=0$ are propagated to $S_{110}(x,z,y)$. 
The unit propagation procedure applied to any instance $I = (V,\{0,1\}, \Cc)\in \CSP(\Gamma_H)$ proceeds as follows: 
\begin{enumerate}
    \item For each  $x\in V$, introduce the unary \textbf{domain} constraint $P_x\subseteq \{0,1\}$. % which contains allowable values for the variable $x$. 
    Initially, set $P_x=\{0,1\}$. % for all $x\in V$. 
    \item \label{horn:domainshrinking} For all  $\ccl{\bs , R}\in \Cc$ (all constraints are at most ternary) and for all $s_i \in \bs$, update $P_{s_i}$ to $P_{s_i}'$ if $P_{s_i}' \subsetneq P_{s_i}$, where
    $$
        P_{s_i} ' = \pi_i\{t \in R : \pi_jt \in P_{s_j} \text{ for all } j\in [k]\}. 
    $$ 
    \item Return to Step~(\ref{horn:domainshrinking}) if any $s_i$ is updated; else 
    terminate. 
\end{enumerate}

The unit propagation corresponds to step~(\ref{horn:domainshrinking}). When the algorithm terminates, there are only two possibilities:
\begin{itemize}
    \item If 
    $P_x=\emptyset$ for some $x$, then $I$ is unsatisfiable.
    \item Else, $I$  has the solution $\psi : V \ar \{0,1\}$:
    \[
        \psi(x)  = \begin{cases}
            a, \text{ if } P_x = \{a\}\\
            0, \text{ otherwise}
        \end{cases}
    \]
\end{itemize}
\end{example}

The notion of 1-minimality is based on the example above. 

\begin{definition}[$1$-minimality]
  A CSP instance $I$ is  \textbf{$1$-minimal} if for all 
  $x\in V$ it contains unary constraint $P_x$, and for all constraints $\ccl{\bs,R_i}$ of $I$ and all $x\in \bs$, we have $\pi_x(R_i\cap\prod_{v\in \bs} P_v) = P_x$. 
\end{definition}

%It s well-known %, e.g., see\cite{16_jl&c_boundedwithhierarchy}, that 
%Any instance $I\in \CSP(\Lc)$ can be transformed into a $1$-minimal instance  $I' = (V,D,\Cc\cup\{P_x: x\in V\})$ over an expanded language such that $I$ and $I'$ are equivalent.   
For completeness,  we provide the following known theorem  since we need to recast it in the automata-theoretic setting:

%by the following polynomial algorithm:

\begin{theorem}\label{Thm-1-minimality} \cite{06_jalgebra_2semilattice,16_jl&c_boundedwithhierarchy}
    There exists an algorithm that transforms every instance $I = (V,D,\Cc)$ to an \textbf{equivalent} $1$-minimal instance $I' = (V,D, \Cc \cup\{P_x: x\in V\})$. The algorithm runs in time $O(|V||D|mr^2)$, where 
     \( r \) is the maximal arity of any constraint relation in \( \Cc \), and  \( m \) is the total number of tuples across all constraint relations.
\end{theorem}
\begin{proof}
Below we present the desired algorithm known as the \textbf{1-minimality algorithm}:
\begin{enumerate}
    \item Initialize $P_x = D$ for all $x\in V$.
    \item  \label{domain_shrinking} For each  constraint $\ccl{(s_1,\ldots,s_k),  R}\in \Cc$ and $i\in [k]$, set
 $
 P_{s_i}'= \pi_i (R\cap\prod_{j \in [k]} P_{s_j}). 
 $
 If  $P_{s_i}'\subsetneq P_{s_i}$, then let  $P_{s_i}=P_{s_i}'$. 
\item If no set $P_x$ where $x\in V$, is updated, then output $I' = (V, D, \Cc\cup\{P_x : x\in V\})$. Else, return to Step $(2)$. 
\end{enumerate}
A proof of correctness can be found in Appendix~\ref{S:Thm-1-minimality}.
\end{proof}

\begin{remark}
    Note that the  sizes of the CSP instances in the theorem above are computed from their explicit representations.
\end{remark}

%% Call the algorithm in the theorem the \textbf{1-minimality algorithm}. 
%% \anote{M.b. name it in the theorem?} (\bakhnote{Done. we will remove this sentence.})
The 1-minimality algorithm in the proof of Theorem~\ref{Thm-1-minimality}, given a CSP instance $I$, produces a 1-minimal CSP instance $I'$
that is equivalent to $I$. Because $I'$ is equivalent to $I$, it is clear that if some $P_x=\emptyset$ in $I'$, then $I$ has no solution since $I'$ does not. In this case we say that the 1-minimality algorithm \textbf{refutes} the instance $I$. However, the converse is not always true. 
%Namely, the condition $P_x\neq \emptyset$, where $x\in V$,  does not imply satisfiability of $I'$ (and hence $I$).

\begin{example}
   Consider the $2$-coloring $\CSP(\neq_2)$. The instance $I_{K_3} = (\{x,y,z\}, \{0,1\}, \{x\neq_2 y, y\neq_2 z, z\neq_2 x\})$ is unsatisfiable because $K_3$ is not $2$-colorable. However, the $1$-minimality algorithm cannot update any domain constraint and returns the trivial $1$-minimal instance with $P_v = \{0,1\}$ for all $v\in\{x,y,z\}$. 
   %% However, $K_3$ is obviously not $2$-colorable.
\end{example}

Now we single out the languages for which the 1-minimality algorithm provides a solution to the CSP instances:

\begin{definition}[Width 1] \cite{barto2009constraint,16_jl&c_boundedwithhierarchy}
    We say that a constraint language  is of \textbf{width 1} if for any unsatisfiable instance  $I$ of the language the $1$-minimality algorithm always refutes $I$.
\end{definition}

\subsection{1-minimality in automata setting}\label{sub:1minauto}

{\bf Inessential expansions:} Before we start this section,  we slightly expand the definition of $\AutCSP$.
\ Consider the language   $\Gamma_{\mathcal A}$, where $\mathcal A$ is an automaton over $D$.   Let 
$\mathcal Q=\{Q_1, \ldots, Q_k\}$ be a fixed set of relations on $D$.
The CSP instances over the language $\Gamma_{\mathcal A}(\mathcal Q)=\Gamma_{\mathcal A} \cup \mathcal Q$ can still be considered as instances of $\AutCSP(\mathcal A)$ because the set $\mathcal Q$ is fixed, and the representations of the relations $Q_1$, $\ldots$, $Q_k$ can be given a priori.
Hence, we call the language 
$\Gamma_{\mathcal A}(\mathcal Q)=\Gamma_{\mathcal A} \cup \mathcal Q$ an   
{\bf inessential expansion of $\Gamma_{\mathcal A}$}, and refer to the CSP instances over $\Gamma_{\mathcal A}(\mathcal Q)$ as the $\AutCSP(\mathcal A)$ instances. This convention is assumed for the rest of this section.

The goal of this subsection is twofold. First,  we prove Theorem \ref{Thm-1-minimality} in the setting of automatic constraint languages $\Gamma_{\Ac}$. 
%Namely, we want to  recast the 1-minimality algorithm for automatic constraint languages. 
The key here is that the $1$-minimality algorithm in  Theorem \ref{Thm-1-minimality} iterates over all possible tuples in the constraint relations as the CSP instances in the theorem are given explicitly. 
%(see Remark after Theorem \ref{Thm-1-minimality}). 
This might blow up the running time exponentially in the size of the representations of the $\AutCSP$ constraints. To avoid this, % exponential blow-up, 
we show an alternative way to establish $1$-minimality for any $\AutCSP$ instance. Second, we show that the satisfiability of an instance $I=(V,D,\Cc,\Ac)$ can be decided in polynomial time if 
the constraint language $\Gamma_\Ac$ has width $1$.

The core of the $1$-minimality algorithm 
is finding domains of individual variables. 
Hence, we first handle projections in the setting of automatic constraint languages $\Gamma_{\Ac}$. 
%Recall that the size of the constraint $\ccl{(s_{1},\ldots, s_{k}),  R_k}$, as explained in Section \ref{SS:size-aut-constraint}, is $O(k)$ (that can be exponentially smaller than $|R_k|$). 

\begin{lemma}\label{lemma:auto_1min_projection}
There exists a polynomial-time algorithm that, given 
    an automaton $\Ac$ over a finite  domain $D$,
    a constraint $\ccl{(s_{1},\ldots, s_{k}),  R_k}$, the domain constraints $P_{s_1}$, $\ldots$, $P_{s_k}$, and $i\in[k]$, computes  
    the projection $\pi_i(R_k \cap \prod_{j\in [k]} P_{s_j})$. 
    The algorithm runs in time $ O(k|D||\Ac|)$.
    \end{lemma}

\begin{proof}
    Note that $\pi_i(R_k \cap \prod_{j\in [k]} P_{s_j}) \subseteq P_{s_i}$. Since $P_{s_i} \subseteq D$, it suffices to check for every element $d\in P_{s_i}$ the  membership of the element $d$ in the set $\pi_i(R_k \cap \prod_{j\in [k]} P_{s_j})$. In the 1-minimality algorithm, this is done by running through every tuple $t\in R_k$. In our case, we invoke the procedure from Lemma \ref{procedure:search_pattern}. We construct the following pattern 
    \[
        \phi(j) = 
        \begin{cases}
            \{d\}, \text{ if } j = i
            \\ P_{s_j}, \text{ otherwise} 
        \end{cases}
    \]
    By using the procedure from Lemma \ref{procedure:search_pattern} with parameters  $\Ac$ and $\phi$, one can find such a desired tuple if it  exists. Note that we iterate over all elements in the set $P_{s_i}$, and each iteration requires $O(k|\Ac|)$ time to answer the membership problem with pattern $\phi$. Thus, it takes time $O(k|P_{s_i}||\Ac|) = O(k|D||\Ac|)$ to compute the full projection.
\end{proof}

 With Lemma \ref{lemma:auto_1min_projection} in hand, we now show that the $1$-minimality algorithm runs in polynomial time for automatic constraint languages.
 The theorem builds an inessential expansion of $\Gamma_{\mathcal A}$. 

\begin{theorem}\label{thm:1-min autcsp}
There exists an algorithm that, given an $\AutCSP$ instance $I=(V,D,\Cc,\Ac)$, constructs an equivalent 
$1$-minimal $\AutCSP$ instance $I'=(V,D,\Cc\cup\{P_x:x\in V\},\Ac)$ in time 
$O(|V||D|^2|\Cc|^2|\Ac|)$.
\end{theorem}

\begin{proof}
Our algorithm follows the same lines as the 1-minimality algorithm in Theorem \ref{Thm-1-minimality}, where the projection procedure uses the process from Lemma \ref{lemma:auto_1min_projection}:
\begin{enumerate}
    \item Initialize $P_x=D$ for all $x\in V$.
    \item \label{step:update} For each constraint $\ccl{(s_1,\ldots,s_k),R}\in\Cc$ and each $i\in[k]$, compute  
    $P'_{s_i}=\pi_i\big(R\cap\textstyle\prod_{j\in[k]}P_{s_j}\big)$
    using Lemma~\ref{lemma:auto_1min_projection}.  \\
    If $P'_{s_i}\subsetneq P_{s_i}$, replace $P_{s_i}$ with $P'_{s_i}$.
    \item If some $P_x$ is updated, then return to Step~(\ref{step:update}); else, output 
    $I'=(V,D,\Cc\cup\{P_x:x\in V\},\Ac)$.  
\end{enumerate}
%Just like in the proof of Theorem \ref{Thm-1-minimality}, 
The algorithm  produces a sequence $I_0$, $I_1$,$\ldots$, $I_w$ of instances while preserving equivalence throughout its execution.  
%Let $I_0 = I$ and $I_w = I'$ denote the initial and final instances, respectively.  Each intermediate instance $I_t$ is obtained from $I_{t-1}$ by shrinking the domain of exactly one variable, while ensuring that no value participating in any solution is removed.  Thus, every update preserves all valid solutions of the previous instance.  
Since each $P_x$ can shrink at most $|D|$ times, the total number of updates is bounded by $O(|V||D|)$.  In addition, for each constraint $\ccl{\bs, R_i}$,  Lemma~\ref{lemma:auto_1min_projection} is applied $|\bs|$ times, giving a total cost of $O(|\bs|^2|D||\Ac|)$.
Processing all constraints in one update step costs $O(|\Cc|^2|D||\Ac|)$  where  $|\Cc| = \sum_{\ccl{\bs, R}\in \Cc} |\bs| $. 
Multiplying the $O(|V||D|)$ bound on the number of updates gives the runtime 
$O(|V||D|^2|\Cc|^2|\Ac|)$.
\end{proof}
    % The equivalence is unchanged, we now argue complexity.
    % The execution of $1$-minimality algorithm can be depicted as a sequence of intermediate instances $\{I_t\}_{t=0}^w$ such that 
    % \begin{itemize}
    %     \item $I_0 = I$ and $I_w = I'$
    %     \item Each update from $I_t$ to $I_{t+1}$ shrinks the domain of exactly one variable.
    %     \item For any $0\le t \le w$, $I_t$ is solution equivalent to $I$.
    % \end{itemize}
    % %  Let $\phi_{t}$ be the domain constraints $\{P_x:x\in V\}$ of $I_t$.
    % % Define the following ordering on domain constraints $\phi_1,\phi_2: V \ar \Pc(D)$ by:
    % % \[
    % %     \phi_1 \le_V \phi_2 \iff \forall x\in V: \phi_1(x) \subseteq \phi_2(x)
    % % \]
    % % Each update is actually strict: $\phi_{t+1} <_V \phi_t$ since exactly one domain constraint would be shrunk. Hence, we conclude that $w$, the number of updates, is bounded by $O(|V||D|)$. 
    
    % For any given constraint $\ccl{(s_1,\ldots,s_k), R}\in \Cc$, the projection procedure \ref{lemma:auto_1min_projection} is invoked for each variable $s_i$. Processing such single constraint costs $O(k^2|D||\Ac|)$. 
    % During a single update, the algorithm iterates over all constraints, so the per-update cost is $O(|\Cc|^2|D||\Ac|)$
    % Consequently, the total running time to establish 1-minimality is $O(|V||D|^2|\Cc|^2|\Ac|)$.

Note that the instance $I'$ in the theorem is over an inessential expansion of $\Gamma_{\mathcal A}$. As an immediate corollary, we get the following:
%of this theorem, we obtain the following result:
%a polynomial time algorithm solving any automatic CSP instance over a constraint language of  width $1$:
\begin{theorem}\label{Thm:automatic-withd1}
    There exists an algorithm, which on any instance $I = (V, D, \Cc, \Ac)$ of width $1$ language $\Gamma_{\Ac}$, decides the satisfiability of $I$ in time $O(|V||D|^2 |\Cc|^2|\Ac|)$.     \qed 
\end{theorem}

We  now apply Theorem \ref{Thm:automatic-withd1} 
to automatic constraint languages 
that admit a semilattice operation, say $\wedge,$ as polymorphisms. 
Then after enforcing $1$-minimality, an automatic CSP instance is satisfiable precisely when every unary domain is nonempty. 
%The next theorem makes this formal.

\begin{theorem}\label{Tm:}
Let $I=(V,D,\Cc,\Ac)$ be an $\AutCSP$ instance such that $\Ac \in \Kc_{\wedge}$, and let $I'=(V,D,\Cc\cup\{P_x:x\in V\},\Ac)$ be an equivalent $1$-minimal $\AutCSP$ instance. 
Then $I$ is satisfiable if and only if every unary domain constraint $P_x$ in $I'$ is nonempty.
%\anote{This can be replaced with a more general statement Víctor Dalmau, Justin Pearson:Closure Functions and Width 1 Problems. CP 1999: 159-173.  However, I'm not sure it's worth it}
\end{theorem}

\begin{proof}
Let us fix $\Ac\in \Kc_{\wedge}$.
If some $P_x=\emptyset$, then variable $x$ has no admissible value in the $1$-minimal instance $I'$, and then $I$ is unsatisfiable.  
Now assume that $P_x\neq\emptyset$ for all $x\in V$. For each  $x\in V$, define $a_x=\bigwedge P_x$, the meet of all elements in $P_x$ under the semilattice operation. Then the function $x\mapsto a_x$ satisfies every constraint of the instance  $I'$. Indeed, with this function, each of the constraints $\langle(s_1,\dots,s_k),R\rangle\in\Cc$ is satisfied by the tuple $(a_{s_1},\dots,a_{s_k})$, since $\wedge$ is a polymorphism of $\Gamma_{\Ac}$.
%Fix a constraint $C=\langle(s_1,\dots,s_k),R\rangle\in\Cc$ and let $T=R\cap\prod_{j\in[k]}P_{s_j}$.   By the definition of the sets $P_{s_j}$, $T$ is nonempty.   Let $t'=\bigwedge_{t\in T}t$, applying the meet coordinatewise.   Since $R$ is closed under $\land$, we have $t'\in R$.   For each $i$, the projection $\pi_i(T)=P_{s_i}$, so $t'(i)=\bigwedge_{t\in T}t(i)=\bigwedge_{p\in P_{s_i}}p=a_{s_i}$.   Thus $(a_{s_1},\dots,a_{s_k})\in R$, and since this holds for every constraint, the mapping $x\mapsto a_x$ is a solution of $I'$ and hence of $I$.
\end{proof}
For this theorem, it is important that in the inessential expansion of $\Gamma_{\mathcal A}$ 
by $\{P_x: x\in V\}$ the predicates $P_x$ are closed under $\wedge$.

\begin{corollary}\label{cll:semi}
There is a polynomial-time algorithm that, given
an $\AutCSP$ instance $I =(V,D,\Cc,\Ac)$ where $\Ac\in \Kc_{\wedge}$, decides if the instance $I$ has a solution. 
\qed 
\end{corollary}

\section{Languages with a majority polymorphism}\label{finitemaj}

In this section, we investigate the class of automata associated with a majority (near-unanimity) polymorphism. Such a class contains many problems (e.g. \textsc{2-SAT}) where the $1$-minimality algorithm in the classical setting is not enough to solve them.

\subsection {Majority polymorphism} \label{SubS:Majority Polymorphisms}

We start with the following definition:
\begin{definition}
An operation $g:D^3\rightarrow D$ is a \textbf{majority operation} if   for 
all $x,y\in D$, \  $g(x,y,y) = g(y,x,y) = g(y,y,x) = y$. 
\end{definition}
Let  $\Ac$ be an automaton  over $D$ such that the language $\Gamma_{\Ac}$ has a majority polymorphism $g$.  
For the operation  $g$ on $D$, consider the set $\Inv(g)$ as in Definition
 \ref{dfn:Pol-Inv}.
%\[\Inv(g) = \{R \subseteq D^n \mid  n\in \Nat, \ g  \  \text{is a polymorphism of $R$} \}.\]
It is well-known that $\CSP(\Inv(g))$ is decidable in polynomial time \cite{Barto14:constraint,16_jl&c_boundedwithhierarchy}. Our goal is to address this result for $\AutCSP(\Ac)$. To do so, we start with the following lemma.

%We start with a lemma that characterizes relations $R$ for which the majority function $g$ is a polymorphism. 

\begin{lemma}
\label{majority_biprojection}
 Let $R\in\Inv (g)$ and $R\subseteq D^n$. Then for all $r\in D^n$ we have  $ r\in R$ if and only if $\pi_{\{i,j\}} r \in \pi_{\{i,j\}} R$ for any $i,j\in [n]$, $i\ne j$. 
\end{lemma}

\begin{proof}
%% \anote{Why do we need this proof? It's a well-known fact.} (\bakhnote{Yes, it is known in the CSP community but not in automata community. If the space is not enough we can move the proof to Appendix.})
The proof in the `only if' direction is obvious. 
We prove the other direction. Let\  $ r = (r_1, r_2, r_3, \ldots, r_n)$ be a tuple such that for all distinct $i,j \in [n]$  we have $(r_i,r_j)\in \pi_{\{i,j\}}R$. We want to show that $ r\in R$. We use induction on the cardinality of  $U\subseteq [n]$, where $|U|\geq 2$, and prove that $\pi_U r \in \pi_U R$. 
\begin{itemize}
        \item \textbf{Base case}:  This is clear when $|U|=2$.
        \item \textbf{Inductive step}: Let $W = \{w_i: 1\le i \le k+1\}\subseteq [n]$ be a subset with $k+1$ elements. Consider the following three $k$ element sets of coordinates: $U_i = W\setminus\{w_i\}$, where $i=1,2,3$. By induction hypothesis, for each $U_i$ there is a tuple $a_i\in R$ such that $\pi_{U_i}a_i = \pi_{U_i} r$. Let us write  
        $a_i=(a_{i,1},\ldots, a_{i,n})$. Note that for all $\ell\in U_i$, \   $r_\ell=a_{i,\ell}$. \ Since $g$ is a polymorphism,  $g(a_1, a_2, a_3) \in R$.
        For $b = g({a}_1, {a}_2, {a}_3)$ we have 
        $b(i)=g(a_{1,i},a_{2,i},a_{3,i})=g(r_i,r_i,r_i)=r_i$, 
        when $i\in U_1\cap U_2\cap U_3$; then, as $g$ is a majority operation, we have
    $b(w_1)=g(a_{1,w_1},a_{2,w_1},a_{3,w_2})=g(a_{1,w_1},r_{w_1},r_{w_1})=r_{w_1}$.
   Similarly $b(w_2)=r_{w_2},b(w_3)=r_{w_3}$. 
        Since $g$ is a polymorphism of the relation $R$, we conclude that  $b$ is the desired tuple, that is, $\pi_W  b = \pi_W r$.
%        \begin{align*}           a_1 &= (\_ \ , r_{w_2}, r_{w_3}, \ldots, r_{w_{k+1}}, \rule{0.5cm}{0.15mm}) \\           a_2 &= (r_{w_1}, \_ \ , r_{w_3}, \ldots, r_{w_{k+1}}, \rule{0.5cm}{0.15mm}) \\        a_3 &= (r_{w_1}, r_{w_2}, \_\ , \ldots, r_{w_{k+1}}, \rule{0.5cm}{0.15mm}) \\      g(a_1,a_2,a_3) &= (r_{w_1}, r_{w_2}, r_{w_3} \ , \ldots, r_{w_{k+1}}, \rule{0.5cm}{0.15mm})        \end{align*}   
\end{itemize}
We proved the lemma.     
\end{proof}

Lemma~\ref{majority_biprojection} implies the following. Let $R$ be an $n$-ary relation that admits the majority operation $g$ as a polymorphism.  For each $i, j\in [n]$, $i<j$, introduce a binary predicate $R_{i,j}$ defined by $\pi_{\{i,j\}}R$.
%such that $R_{i,j}(a,b)$ if and only if $(a,b)\in \pi_{i,j}(R)$
Then the relations $R_{i,j}$ determine $R$ in the following sense:
\begin{corollary} \label{Cor: Inv(Maj) are bijunctive}
 Let $R\subseteq D^n$ be a relation invariant under $g$. For any tuple $(a_1,\ldots, a_n)\in D^n$, the following is true: 
 \[
    R(a_1,\ldots, a_n) \Longleftrightarrow \bigwedge_{1\le i<j\leq n} R_{i,j}(a_i, a_j)
 \]
\end{corollary}
    
Assume that $R_n\in \Gamma_{\Ac}$. %; so,  $R=R_n$ for some $n$.  
For a pair $(a,b)\in D$ and the relation $R_{i,j}$, let us consider a partial function $\tau: [n] \partialar D$ such that $\tau(i)=a$, $\tau(j)=b$, and $\tau(k)$ is undefined for all other $k\in [n]$. By  Lemma \ref{procedure:search_pattern}, we can find if the partial function $\tau$ has an extension in $R$ in time $O(n |\Ac|)$.  This implies that we can compute the relation $R_{i,j}$ in time $O(n|\Ac| |D|^2)$ by trying every possible pair $(a,b)\in D^2$. Hence:
%Therefore, the running time of the algorithm from the corollary above can be compressed. Namely, the following is true:

\begin{corollary} \label{Cor:R(i,j)-raw}
 There exists an algorithm that, 
 given a finite automaton $\Ac \in \Kc_g$ and an integer $n$, outputs the relations $R_{i,j}$ such that
 for all $(a_1, \ldots, a_n)\in D^n$ we have the following:
 \begin{center}{
 $R_n(a_1,\ldots, a_n) \Longleftrightarrow \bigwedge_{i<j\leq n} R_{i,j}(a_i, a_j)$.}
 \end{center}
 The algorithm runs in time $O(n^3|D|^2||\Ac|)$.  \qed 
\end{corollary}

%Now we  prove the following theorem:

\begin{theorem}\label{Thm:Maj-Transformation} There exists an algorithm that, given an $\AutCSP$ instance  $I=(V,D,\Cc,\Ac)$ where $\Ac \in \Kc_{g}$,  constructs a CSP instance  $I'=(V,D,\{\langle(x,y),P_{xy}\rangle : x,y\in V\})$ such that
\begin{enumerate}
\item Each $P_{xy} \subseteq D^2$ belongs to $\Inv(g)$.
\item A function $\psi:V\rightarrow D$ satisfies $I$ iff $\psi$ satisfies $I'$.
\item The algorithm runs in time $O(|\Cc|^3|\Ac||D|^2)$.
\end{enumerate}
\end{theorem}\label{redutiontoCSP}
%\begin{theorem} There exists an algorithm that, given an automatic CSP instance  $I=(V,D,\Cc,\Ac)$,  constructs an equivalent CSP instance  $I'=(V,D,\{\langle(x,y),P_{xy}\rangle : x,y\in V\})\in\CSP(\Inv(g)),$ where for each $x,y\in V$, \[ P_{xy}=\bigcap\{\pi_{i,j}R : \langle(s_1,\dots,s_k),R\rangle\in\Cc,\; s_i=x,\; s_j=y,\; i\ne j\}. \] Moreover, the algorithm runs in time $O(|\Cc|^3|\Ac||D|^2)$.\end{theorem}
\begin{proof}
For each $x,y\in V$, $P_{xy}$ is defined as follows: 
\[ 
P_{xy}=\bigcap_{\langle(s_1,\dots,s_k),R\rangle\in\Cc,\; s_i=x,\; s_j=y,\; i\ne j} R_{i,j}. 
\]
For the first part, clearly $P_{xy}\in Inv(g)$. \ 
For the second part, let $\psi: V \ar D$ be a function. By Corollary \ref{Cor: Inv(Maj) are bijunctive}, we have the following:
%following series of equivalent conditions of $\psi$ being a solution of $I$:
    \begin{align*}
        \bigwedge_{\ccl{\bs,R}\in \Cc}\psi(\bs)\in R 
        \iff \bigwedge_{\ccl{\bs,R}\in \Cc} \bigwedge_{(s_i, s_j)\in \bs^2} (\psi(s_i),\psi(s_j))\in R_{i,j}   \\ 
        \iff \bigwedge_{(x,y)\in V^2}\bigwedge_{\substack{\ccl{(s_1,\ldots,s_k),R}\in \Cc \\ s_i  = x, s_j = y}}  (\psi(x),\psi(y))\in R_{i,j} \\ 
    \end{align*}
%Note that the first equivalence follows from.  Now we analyze the running time of the algorithm.  
For each constraint $\langle(s_1,\dots,s_k),R\rangle\in\Cc$  
and each distinct pair $(i,j)\in[k]^2$,  
we compute $R_{i,j}\subseteq D^2$ by enumerating all $(a,b)\in D^2$  
and testing membership using Corollary \ref{Cor:R(i,j)-raw}.  
Each projection 
% requires to iterate over all pairs in $D^2$ and 
costs $O(k|D|^2|\Ac|)$ time,
and there are $O(k^2)$ such pairs per constraint,  
giving a total per-constraint cost of $O(k^3|\Ac||D|^2)$.  
Iterating over all constraints gives overall complexity $O(|\Cc|^3|\Ac||D|^2)$.
% For each constraint, we have $O(k^2)$ many pairs of indices to calculate the projection. Thus, each constraint $\ccl{(s_1,\ldots,s_k), R}$ takes time $O(k^3|\Ac||D|^2)$. The algorithm iterates over all constraints and 
Thus, $I'$ can be constructed in polynomial time and is equivalent to $I$.
\end{proof}

%The reduction shows that AutCSPs whose constraint language admits a majority polymorphism can be decided in polynomial time. 

Since the constraints in the instance $I'$ in Theorem~\ref{Thm:Maj-Transformation} are binary, we can use the explicit representation of these constraints:
%without a risk of blow-up in size. 

%Therefore, we can use the standard approach to solving such instances:

%. We now present a known polynomial-time algorithm \cite{barto2009constraint} that decides satisfiability of  CSP instances that contain only binary constraints from $\operatorname{Inv}(g)$. 

\begin{theorem}\label{Thm:Maj-Transformation} \cite{barto2009constraint}
There is an algorithm that, given a binary CSP instance $I =(V,D,\{\langle(x,y),P_{x,y}\rangle : x,y\in V\})$ where all $P_{x,y} \subseteq D^2$ are in  $\Inv(g)$, decides the satisfiability of $I$ in time $O(|V|^3  |D|^3)$.
\end{theorem}

\begin{proof}Iteratively prune the relations $P_{x,y}$ as follows: for every triple $(x,y,z)$ and every pair $(a,b) \in P_{x,y}$, check whether there exists $c \in D$ such that $(a,c) \in P_{x,z}$ and $(b,c) \in P_{y,z}$. If no such $c$ exists, delete $(a,b)$ from $P_{x,y}$. If any $P_{x,y}$ becomes empty, report ``unsatisfiable''; otherwise, after no more deletions are possible, report ``satisfiable''. \ This  process examines all triples and all values, taking time $O(|V|^3  |D|^3)$. Its correctness is proved in Appendix~\ref{S-Thm:Maj-Transformation}. 
\end{proof}
Thus, we obtain the following corollary:

\begin{corollary}\label{cor:majority}
There exists a polynomial-time algorithm that, given
an $\AutCSP$ instance $I =(V,D,\Cc,\Ac)$ where $\Ac\in \Kc_{g}$, decides if the instance $I$ has a solution. The algorithm runs in time $O(|\Cc|^3|\Ac||D|^2+|V|^3|D|^3)$ where $|\Cc| = \sum_{\ccl{\bs, R}\in \Cc} |\bs|$. \qed 
\end{corollary}

\subsection{Near-Unanimity Polymorphisms} 
The majority is a special case of near-unanimity operation. 
A map $g:D^k\rightarrow D$ is a \textbf{near-unanimity (NU)} operation if   for 
all $x,y\in D$:
\[ 
g(x,y,y,\ldots, y) = g(y,x,y \ldots,y)=\ldots = g(y,y,\ldots,y,x) = y.
\] 
One can easily  generalize the results of Section \ref{SubS:Majority Polymorphisms} on majority polymorphisms to near unanimity polymorphisms. Indeed, first,  Lemma \ref{majority_biprojection} holds true
for NU polymorphisms:

\begin{lemma}[\cite{Jeavons97:closure}]\label{lemma:nu-projection_k}
        Let $g$ be a $k$-ary NU polymorphism  of a relation  $R\subseteq D^n$. Then for all $ r\in D^n$ we have  $ r\in R$ if and only if $\pi_U r \in \pi_U R$ for any subset $U\subseteq [n]$ of cardinality $k-1$. \qed
\end{lemma}
\noindent 
Second, Theorem \ref{thm:1-min autcsp} also holds true for NU polymorphisms:
\begin{theorem} \label{Thm: NU reduction}
    Let $g$ be a NU operation. There is an algorithm that, given an $\AutCSP$ instance  $I=(V,D,\Cc,\Ac)$ where $\Ac \in \Kc_{g}$,  builds a CSP instance  $I'=(V,D,\{\ccl{\ub,P_{\ub}} : \ub\subseteq V, |\ub| = k-1\})$ such that
\begin{enumerate}
\item Each $P_{\ub} \subseteq D^{k-1}$ belongs to $\Inv(g) $.
\item A map $\psi:V\rightarrow D$ satisfies $I$ if and only if $\psi$ satisfies $I'$.
\item The algorithm runs in time $O(|\Cc|^{k}|\Ac||D|^{k-1})$. \qed
\end{enumerate}
\end{theorem}

The proof is in Appendix \ref{Appendix: NU reduction}. Finally, we invoke the known fact that  $\CSP(\Inv(g))$ is polynomial-time solvable  \cite{Barto14:constraint}: % Thus, we have the following corollary.
\begin{corollary}\label{cor:NU}
Let $g$ be a $NU$ operation. There exists a polynomial-time algorithm that, given
an $\AutCSP$ instance $I =(V,D,\Cc,\Ac)$ where $\Ac\in \Kc_{g}$, decides if the instance $I$ has a solution.
 \qed
\end{corollary}

\section{Dichotomy for Boolean AutCSPs}\label{Thm:boolean-dichotomy}
We now combine the results above to obtain a full complexity dichotomy for AutCSPs over the Boolean domain.
\begin{theorem}\label{Thm:Automata-dichotomy}
    Let $\Ac$ be an automaton over the Boolean domain $\{0,1\}$. 
    Then $\AutCSP(\Ac)\in P$ if and only if \  $\Gamma_{\Ac}$ admits one of the following polymorphisms: $\mathbf{0}$, $\mathbf{1}$, $\land$, $\lor$,
    $Maj$, $Minor$. Otherwise the $\AutCSP(\Ac)$  is NP-complete. 
\end{theorem}\label{{Thm:Boolean-Dichotomy}}
\begin{proof}
Assume that the automatic constraint language $\Gamma_{\Ac}$ admits none of the listed operations $\mathbf{0}$, $\mathbf{1}$, $\land$, $\lor$,
    $Maj$, and $Minor$  as its polymorphisms.  
For each of the operations $f$ from the set $\{\mathbf{0}, \mathbf{1}, \land, \lor, Maj, Minor\}$ build an automaton $\Ac_{f}$ from Theorem \ref{Thm:deciding-polymorphism}. Let us assume, for simplicity, that $f$ is either $\land$ or $\lor$. Since $f$ is not a polymorphism of $\Gamma_{\Ac}$, there is a triple
of strings $(x_f, y_f, z_f)$ accepted by $\Ac_{f}$.  Hence, $x_f, y_f$ are accepted by $\Ac$, but $f(x_f, y_f)=z_f$ is not accepted by $\Ac$. 
Let $n_f=|x_f|$ and $R_{n_f} \;=\; \mathcal L(\Ac) \cap D^{n_f}$. 
Hence, the collection of relations
$$
\Gamma' \;=\; \{\, R_{n_f} \mid f \in \{\mathbf{0},\mathbf{1},\land,\lor,Maj,Minor\}\,\}
$$
% the finite language  $\Lc'$ consisting of relations
% $$
% R_{n_f}, \ \textit{where $f\in \{\mathbf{0}, \mathbf{1}, \land, \lor, Maj, Minor\}$}
% $$
forms a finite language over the Boolean domain that admits none of 
 the listed operations as its polymorphisms.
%,then there exists a finite sub-language $\Lc\subseteq \Lc_{\Ac}$ witnessing this fact, which can be effectively found by Theorem~\ref{Thm:deciding-polymorphism}.
% By Theorem~\ref{Thm:Boolean}, 
By Schaefer's Dichotomy Theorem \cite{SchaeferDic},  
$\CSP(\mathcal{L'})$ is NP-complete. Indeed, given any instance 
$I=(V,D,\{\ccl{s_1,R_1},\dots,\ccl{s_m,R_m}\})$
of $\CSP(\mathcal{L'})$,
we construct the instance 
$I'=(V,D,\{\ccl{s_1,R_1},\dots,\ccl{s_m,R_m}\}) \in \AutCSP(\Ac)$ in linear time.
% that is, we simply reuse the same set of variables, domain,
% and constraints and attach the fixed automaton~$\Ac$.
Clearly $I$ is satisfiable if and only if $I'$ is satisfiable.
Thus $\AutCSP(\Ac)$ is NP-hard. By Proposition \ref{Prop:$NP$-complete}, $\AutCSP(\Ac)$  is  in NP. Hence,  $\AutCSP(\Ac)$ is NP-complete.

If $\Gamma_{\Ac}$ admits one of the listed automatic polymorphisms, then tractability follows from the corresponding results:
semilattice operations ($\land$, $\lor$) by Theorem~\ref{Thm:Booleam-AND-Polymporphism};
$Minor$ by Theorem~\ref{thm:2minor_translation};
$Maj$ by Theorem~\ref{redutiontoCSP};
and the unary operations $\mathbf{0}$, $\mathbf{1}$ by trivial evaluation of constant assignments.
Hence $\AutCSP(\Ac)$ is in P.
\end{proof}

\noindent 
 Now combine  Theorem \ref{Thm:boolean-dichotomy} and Corollary \ref{Cor:Pol-check}
 for the next statement:
 
\begin{corollary}\label{Cor:deciding-dichotomy}
There is a polynomial-time algorithm with the following properties:
\begin{enumerate}
    \item Given a finite automaton $\Ac$  over $D=\{0,1\}$ as input,  the algorithm decides if  $\AutCSP(\Ac)$ is in P.
    \item If $\AutCSP(\Ac)$ is not in P, then the algorithm  produces a finite set $\Gamma'$ of relations  from $\Gamma_{\Ac}$ witnessing that
    $\AutCSP(\Ac)$ is NP-complete.
    \item If $\AutCSP(\Ac)$ is in P, then the algorithm produces a decision process for $\AutCSP(\Ac)$ running in polynomial time on instances $I\in \AutCSP(\Ac)$ and the automaton $\Ac$.  \qed 
\end{enumerate}

\end{corollary}

\section{Conclusion}
%We introduced the framework for Automatic Constraint Satisfaction Problems (AutCSPs), extending classical CSPs.  We developed the algebraic theory of automatic polymorphisms, showed that their existence is decidable in polynomial time, and gave polynomial-time algorithms for AutCSPs whose languages admit semilattice or minority polymorphisms on the Boolean domain, and we extended the semilattice and majority cases to arbitrary finite domains. These results yield a dichotomy theorem for Boolean AutCSPs, classifying every problem as either polynomial-time solvable or NP-complete.

Several questions remain open. First, while classical CSPs over finite domains are tractable when the language admits a Mal’tsev polymorphism~\cite{Bulatov2002TractableMaltsev,Bulatov2006SimpleMaltsev}, it is unknown if  Mal’tsev polymorphisms imply tractability for AutCSPs. Second, we  lack an algorithm that converts an AutCSP instance into an equivalent $(2,3)$-minimal instance \cite{barto2009constraint,16_jl&c_boundedwithhierarchy}, a central tool in classical relational width and local-consistency frameworks. Finally, we aim to extend our results beyond the Boolean case by establishing a full dichotomy for AutCSPs over the 3-element domain, further clarifying the structure and complexity landscape of automatic constraint languages.

%\section{Inv-Pol   Revisited} Given a constraint language $\Lc$, any constraint language $\Gamma$, which comes from $\mathrm{Inv}(\mathrm{Pol}(\Lc))$, the relational clone generated by $\Lc$, has the same complexity as original $\Lc$\cite{Bulatov05:classifying}. Therefore, it's reasonable to replace CSP$(\Lc)$ with any $\CSP(\Gamma)$ as long as $\Pol(\Lc) = \Pol(\Gamma)$.  \xynote{I think we need to prove the following theorem to justify that it's reasonable to only consider AutCSP up to the polymorphism it has. \begin{theorem} Let $\Ac,\Bc$ be two automata with $\Pol(\Ac) = \Pol(\Bc)$. Then AutCSP$(\Ac)$ is polynomial time equivalent to AutCSP$(\Bc)$. \end{theorem}}

\bibliographystyle{ACM-Reference-Format}
\bibliography{Reference}

%% \section{Appendix}

\  \\

\appendix
\section{Proof of Corollary \ref{Cor:Pol-check}}\label{S:Pol-Check}
\begin{corollary 4.5}
Let $\Ac$ be a finite automaton over $D$. Then:
\begin{enumerate}
    \item It is decidable in polynomial time whether  or not $\Gamma_{\Ac}$ admits a Sigger's operation as its polymorphism.
    \item Moreover, if no Sigger's operation is a polymorphism of $\Gamma_{\Ac}$, then we can find in time $O(|\Ac|^5)$ a finite sub-language $\Gamma \subseteq \Gamma_{\Ac}$ such that the $CSP(\Gamma)$, and therefore $\AutCSP(\Ac)$, are NP-complete. 
\end{enumerate}  
 \end{corollary 4.5}
 \begin{proof}
Recall that the domain $D$ is fixed. Therefore we can pre-compute  
all Sigger's operations: $g_1$, $\ldots$, $g_k$. 
%Hence, this function list can be assumed to be given. 
Using Theorem \ref{Thm:deciding-polymorphism}, for each $g_i$, $i=1, \ldots, k$, we check if $g_k$ induces a polymorphism of $\Gamma_{\Ac}$. This proves the first part.

 For the second part, for each $i$, where $i=1, \ldots, k$ we can find a relation $R_{n_i}\in \Gamma_{\Ac}$ such that $g_i$ is not a
 polymorphism for $R_{n_i}$. Then none of the Sigger's operations is  a polymorphism of the finite language  $\Gamma=\{R_{n_1}, \ldots, R_{n_k}\}$. Hence, by  Theorem \cite{DBLP:conf/focs/Bulatov17,DBLP:conf/focs/Zhuk17}, the problem $\CSP(\Gamma)$ is NP-complete. Hence, 
 % $\CSP(\Gamma_{\Ac})$
 $AutCSP(\Ac)$ 
 is NP-complete.
 \end{proof}

\section{Proof of Lemma \ref{lemma:finite_field_extraction}}\label{Appendix: GF(q) extraction}
\begin{lemma 5.8}
   There exists an algorithm, which given an arity $n$ and an automaton $\Ac\in \Kc_{f_q}$, constructs in time $O(n^3|\Ac|)$a system of linear equation $Mx = b$ such that $r\in R_n$ if and only if $Mr^\intercal = b$. 
\end{lemma 5.8}
\begin{proof}
    Note that $R_n$ is an affine subspace of space $GF(q)^n$. Hence, the relation $R_n$ is represented as $r+ S$, where $r$ is an arbitrary vector from $R_n$ and $S\le GF(q)^n$ is some subspace. Finding $r$ is rather easy by testifying extendability of the empty pattern $\psi_\emptyset:[n] \partialar GF(q)$. to find the subspace $S$, we construct the basis $\Bc$ of $S$. When both $r$ and $S$ are found, the matrix representation $Mx = b$ is constructed by taking a matrix $M$ with kernel $S$ and a vector $b = Mr^\intercal$.

    We now outline the procedure for finding the basis $\Bc$. Suppose that the subspace $S$ has dimension $m$. The construction of $\Bc$ is a series of bases: $\emptyset = \Bc_0 \subsetneq \Bc_1 \subsetneq \Bc_2\subsetneq \ldots \subsetneq\Bc_m = \Bc$. When $\Bc_{i}$ is given at $i$-th step, we try to find a new vector $b_i\in S$ independent of $\Bc_i$. Then $\Bc_{i+1}$ is simply $\Bc_{i}\cup\{b_i\}$. If no such vector $b_i$ exists, obviously we reach the end of construction, that is $i=m$. 

    To simplify the search for a candidate vector $b_i$, we assume that each $\Bc_i$ is normalized to row canonical form. Let $L_i\subseteq [n]$ be the columns of the leading $1$s of basis $\Bc_i$. Since $\Bc_i$ is in row canonical form, a candidate vector $b_i$ exists if and only if there is some non-zero vector $b\in S$ such that $b(j) = 0$ for all $j\in L_i$.
    Moreover, because $b\in S$ is non-zero, there must be a column $p\in [n]\setminus L_i$ such that $b(p) \neq 0$.  Given $L_i$, we can detect such a candidate by testing the extendability of the following patterns. For each  column  $p\in [n]\setminus L_i$, define a pattern $\phi_{p,L_i}:[n]\partialar GF(q)$ by 
         \[
        \phi_{p,L_i}(k) = \begin{cases}
            r(k) + 1 , k=p \\ 
            r(k) , k\in L_i \\
            \text{undefined}, \text{ otherwise}
        \end{cases}
        \] where addition is the group addition  of $GF(q)$.
 There exists $x\in R$ extending $\phi_{p,L_{i}}$ if and only if $b = x - r\in S$ satisfies $b(k) = 0$ for all $k\in L_i$ and $b(p) = 1$, in which case we may take $b_i = b$.
\end{proof}

\section{Proof of Lemma \ref{procedure:search_pattern}}\label{S:procedure:search_pattern}
\begin{lemma 6.2}
    There exists an algorithm that, for a given  pattern $\phi: [n] \ar \Pc(D)$ and an automaton $\Ac$, decides if there exists a tuple $t\in \mathcal  L(\Ac)\cap \prod_{i=1}^n\phi(i)$ and, if so, constructs such a tuple $t$. The algorithm runs in time  $O(n|\Ac|)$. 
\end{lemma 6.2}

\begin{proof}
    Consider the transition graph $G(\Ac) = (S, \ar_{\Delta})$ of the automaton $\Ac$, where  $s_1\ar_{\Delta} s_2$ if and only if $s_2 = \Delta(s_1,a)$ for some $a\in D$.  One can compute a walk $\{s_i\}_{i=0}^n$ with the following properties:
    \begin{enumerate}
        \item $s_0 = q_0$ and $s_n \in F$
        \item For $i\in dom(\phi)$, $s_i = \Delta(s_{i-1}, a) $ for some $a\in \phi(i)$.
    \end{enumerate}
 This walk can easily be computed by a standard Breadth-First-
Search algorithm with running time $O(n|\Delta|)= O(n|\Ac|)$.
\end{proof}

\section{The rest of the proof of Theorem \ref{Thm-1-minimality}}\label{S:Thm-1-minimality}
\begin{theorem 6.5}
    There exists an algorithm that transforms every instance $I = (V,D,\Cc)$ to an \textbf{equivalent} $1$-minimal instance $I' = (V,D, \Cc \cup\{P_x: x\in V\})$. The algorithm runs in time $O(|V||D|mr^2)$, where 
     \( r \) is the maximal arity of any constraint relation in \( \Cc \), and  \( m \) is the total number of tuples across all constraint relations.
\end{theorem 6.5}
The process produces a sequence of instances $I_0,I_1,\dots,I_w$,  
where $I_0=I$ and each $I_{t+1}$ is obtained from $I_t$ by shrinking the domain of exactly one variable.  
We show that each update preserves all satisfying assignments.  

Assume that $\psi$ satisfies $I_{t+1}$. Since $P_{s_i}'\subseteq P_{s_i}$, $\psi$ also satisfies $I_t$.  
Conversely, assume $\psi$ satisfies $I_t$.  
For any constraint $\langle(s_1,\dots,s_k),R\rangle$,  
we have $(\psi(s_1),\dots,\psi(s_k))\in R$ and $\psi(s_j)\in P_{s_j}$ for all $j$.  
By definition of $P_{s_i}'=\pi_i(R\cap\prod_j P_{s_j})$,  
it follows that $\psi(s_i)\in P_{s_i}'$, so $\psi$ also satisfies $I_{t+1}$.  
Thus every satisfying assignment of $I_t$ remains valid after the update.  
By induction, the final instance $I'$ is therefore equivalent to $I$.

Note that we can bound the number of updates by $ O(|V||D|)$. 
Let \( r \) be the maximal arity of any constraint relation in \( \Cc \), and let \( m \) be the total number of tuples across all constraint relations.
For a given constraint $\ccl{(s_1,\ldots, s_k), R}\in \Cc$, calculating $P_{s_i}'$ for single variable $s_i$ takes time $O(|R|k) = O(|R|r)$. Hence, it takes $O(|R|k^2) = O(|R|r^2)$  to process one such constraint. During each update, the algorithm iterates over every constraint from $\Cc$ and runs in time $O(mr^2)$. Hence, the overall running time 
%of $1$-minimality algorithm is bounded by 
is $O(|V||D|mr^2)$.

\section{The rest of the proof of Theorem \ref{Thm:Maj-Transformation}} \label{S-Thm:Maj-Transformation}
\begin{theorem 7.6}
There exists an algorithm that, given a binary CSP instance $I =(V,D,\{\langle(x,y),P_{x,y}\rangle : x,y\in V\})$ where $P_{x,y} \subseteq D^2$ belongs to $\Inv(g)$, decides the satisfiability of $I$ in time $O(|V|^3 \cdot |D|^3)$.
\end{theorem 7.6}
If some $P_{xy}$ is emptied, no global assignment extending all pairs can exist. Conversely, if all $P_{xy}$ remain nonempty, then by the following clique-extension argument, local consistency implies global satisfiability. 
Given a triangle $(x,y,z)$ with pairs $(a,b)$, $(a,c)$, and $(b,c)$, one can extend this assignment to any fourth variable $w$ as follows: for each of the three edges, take a corresponding witness value $d_1,d_2,d_3$ for $w$; define $d = g(d_1,d_2,d_3)$. 
Since $bd_1 \in P_{yz}$, there exists $a'$ such that $a'd_1 \in P_{xz}$.
Hence, we have $(a,d)=(g(a',a,a),g(d_1,d_2,d_3)) \in P_{xy}$.
Similarly, we have the value $d$ pairs consistently with $a,b,c$, thus extending the triangle to a 4-clique. 
In a similar way, we can show that every $4$–clique extends to a $5$–clique, and so on. Summarizing, every edge extends to a $|V|$–clique – a solution to the instance.

\section{Proof of Theorem \ref{Thm: NU reduction}}\label{Appendix: NU reduction}
\begin{theorem 7.9} 
    Let $g$ be a $NU$ operation. There exists an algorithm that, given an $\AutCSP$ instance  $I=(V,D,\Cc,\Ac)$ where $\Ac \in \Kc_{g}$,  constructs a CSP instance  $I'=(V,D,\{\ccl{\ub,P_{\ub}} : \ub\subseteq V, |\ub| = k-1\})$ such that
\begin{enumerate}
\item Each $P_{\ub} \subseteq D^{k-1}$ belongs to $\Inv(g) $.
\item A map $\psi:V\rightarrow D$ satisfies $I$ iff $\psi$ satisfies $I'$.
\item The algorithm runs in time $O(|\Cc|^{k}|\Ac||D|^{k-1})$.
\end{enumerate}
\end{theorem 7.9}
\begin{proof}
    By lemma \ref{lemma:nu-projection_k}, we establish the following equivalences:
    \begin{align*}
        \bigwedge_{\ccl{\bs,R}\in \Cc}\psi(\bs)\in R 
        \iff \bigwedge_{\ccl{\bs,R}\in \Cc} \bigwedge_{\ub \subseteq \bs,|\ub| =k-1} \psi(\ub)\in \pi_\ub R   \\ 
        \iff \bigwedge_{\ub\subseteq V,|\ub| = k-1}\bigwedge_{\substack{\ccl{\bs,R}\in \Cc \\ \ub\subseteq \bs }}  \psi(\ub)\in \pi_\ub R \\ 
    \end{align*}

    Therefore, we can introduce the following $k-1$-ary relation :
    \[
        P_\ub = \bigcap_{\ccl{\bs,R}\in \Cc, \ub \subseteq \bs} \pi_\ub R
    \]

    For any scope $\ub$, $P_\ub$ has $g$ as its polymorphism since projections inside does. These relations $P_\ub$ would be constraint for our $I'$ and first condition has been satisfied.

    Then, we can construct the following instance $I' = (V,D,\Cc')$ equivalent to original $I$:
    \[
        \Cc' = \{\ccl{\ub, P_\ub} : \ub \subseteq V, |\ub| = k-1\}
    \]
    It is not hard to see that $I'$ already satisfied second condition. 

   Let us analyze the complexity of constructing $I'$. For a single constraint $\ccl{\bs,R}\in \Cc$, we compute the projection $\pi_\ub R$ for each scope $\ub$ of size $k-1$. There are $|\bs|^{k-1}$ such scopes and each can be processed in time $O(|D|^{k-1}|\Ac||\bs|)$. Therefore, it takes time $O(|\bs|^{k}|\Ac||D|^{k-1})$ to complete single constraint $\ccl{\bs,R}\in \Cc$. It thus takes time $O(|\Cc|^{k}|\Ac||D|^{k-1})$ in total to construct the instance $I'$ for a given $I = (V,D,\Cc,\Ac)$.
\end{proof}

\newpage

\tableofcontents

\end{document}